\RequirePackage{fix-cm}
\documentclass[smallextended]{svjour3}       
\smartqed  
\usepackage{amssymb}
\usepackage{amsmath}
\makeatletter  
\spn@wtheorem{assumption}{Assumption}{\bfseries}{\rmfamily}
\makeatother   
\usepackage{algorithmic}
\usepackage{algorithm}
\usepackage{booktabs}
\usepackage{siunitx}
\usepackage{multirow}
\usepackage{caption,subcaption}
\usepackage{graphicx}
\graphicspath{{fig/}}
\def\R{\mbox{R}}
\def\F{\mbox{F}}
\def\RE{\mbox{RE}}
\def\MSE{\mbox{MSE}}
\def\PSNR{\mbox{PSNR}}

\def\SSIM{\mbox{SSIM}}

\def\sgn{\mbox{sgn}} 
\begin{document}
	\begin{sloppypar}

\title{A hybrid statistical sampling and iterative regularization method in sparse-view computed tomography
\thanks{This research was funded by National Natural Science Foundation of China (NSFC) grant number 12271312 and Shandong Provincial Outstanding Youth Fund (Grand No. ZR2018JL002)}}



\author{Huiying Li         \and
        Yizhuang Song 
}


\institute{Huiying Li, \at
              Mathematics and Statistics, Shandong Normal University, Jinan, Shandong, P. R. China \\
           \and
           Yizhuang Song (Corresponding author) \at
           Mathematics and Statistics, Shandong Normal University, Jinan, Shandong, P. R. China
              \email{ysong@sdnu.edu.cn} 
}

\date{Received: date / Accepted: date}

\maketitle

\begin{abstract}
Sparse-view computed tomography (CT) is an effective method to reduce the radiation exposure in medical imaging. To reduce the severe streaking artifacts that occur in reconstructed images due to violation of the Nyquist/Shannon sampling criterion, regularization is widely used to minimize the cost function. However, the iterative methods may lead to the accumulation and propagation of errors, which adversely affects the restoration of image details and textures. In this paper, we propose a hybrid model that integrates statistical sampling with iterative regularization to simultaneously shorten the sampling time and enhance the reconstruction quality. The proposed method is validated using three datasets: the Shepp-Logan phantom, the actual walnut X-ray projections provided by the Finnish Inverse Problems Society, and the clinical lung CT images.
\keywords{statistical sampling \and regularization \and hybrid model\and sparse-view CT}
\subclass{65R32 \and 65K10 \and 68U10}
\end{abstract}

\section{Introduction}\label{sec:intro}
Computed Tomography (CT)\cite{CT} is a non-destructive imaging technique that exploits the differential absorption of X-rays by various tissues and organs to visualize the internal structures of the human body.  Furthermore, the X-ray absorption capacities of diseased tissues and normal tissues are different. The physical quantity that characterizes a tissue's ability to absorb X-rays is the attenuation coefficient $\mu$. The CT mathematical principle is to detect the residual intensity of X-rays after attenuation (projection data) by the imaging object and then use the inverse Radon  transform\cite{Radon} to reconstruct the attenuation coefficient $\mu$. Visualizing the attenuation coefficient $\mu$ can assist doctors in diagnostic decision-making.

CT provides high-resolution images of tissues and organs. However, due to the high-energy attribute of X-rays, traditional CT has the risk of cancer\cite{Hall2008}. Therefore, low-dose CT has remained a prominent research focus in this field. Our work focuses on sparse-view CT, a main method of low-dose CT. Sparse sampling violates the Nyquist/Shannon sampling criterion, resulting in severe streaking artifacts in the reconstructed attenuation coefficient images\cite{Hu2014}. Such artifacts can confuse the artifacts with tissue lesions, potentially causing misdiagnosis by doctors. Therefore, Reducing streaking artifacts is a key scientific issue in sparse-view CT reconstruction.

Based on the Algebraic Reconstruction Technique\cite{CT} (ART), we add the appropriate regularization terms into the cost function to reduce streaking artifacts. The commonly used functional minimization methods containing regularization terms can be classified into direct matrix inversion method\cite{CT}, iteration-based minimization method\cite{CT}, statistical calculation method\cite{SomersaloBook}, and neural network method\cite{AI}, etc. Among these, the direct matrix inversion method is often used in the minimization of the energy functional containing a Tikhonov regularization term\cite{CT}. 
This method has high computational efficiency. However, the Tikhonov regularization term tends to blur boundaries of internal tissues, which affects doctors' judgment of the properties of lesion tissues. The iteration-based minimization methods are often used to minimize the energy functional, which contains regularization terms specifically designed to preserve boundary features.
Regularization terms of this type mainly include Total Variation\cite{CT} (TV), high-order TV\cite{Niu2014}, L1/L2\cite{L1/L2}, and Nonlinear Weighted Anisotropic TV\cite{song,self} (NWATV) and so on. However, the iterative method often leads to the accumulation of calculation errors or noise, which may result in the non-convergence of the iterative process and the blurring of small abnormal tissues. The statistical calculation method uses statistical sampling techniques to minimize the energy functional. This method can effectively prevent small anomalies from being blurred.
However, in actual medical imaging, this method is limited used due to its low calculation efficiency. The neural network method learns a deep neural network by using labeled data to describe the mapping relationship between the input (the initial energy $I_0$ of X-rays, the projection data $I(\theta,s)$) and the output $\mu$. While this method can establish a relatively accurate mapping relationship based on the sufficient high-quality data, the neural network is a black box, the performance depends on the quality of the selected training set data, and its stability is not high\cite{AI2}.

To reduce the streaking artifacts, we construct an energy functional including a fidelity term and boundary-preserving regularization terms. For the optimization problem of this functional, we propose a novel hybrid optimization method that integrates the fast convergence of the iterative minimization algorithm with the edge-preserving capability of the statistical calculation method. 
The specific process includes: Firstly,  locate the suspected lesion area (i.e., the region of interest, abbreviated as ROI) based on experience and extract the corresponding projection data; Subsequently, the attenuation coefficients of ROI are statistically sampled using the Markov Chain Monte Carlo (MCMC) method. 
Then, the conditional mean estimation (CM) of the sampling results is imported as prior information into the energy functional to be minimized. The advantages of this method are as follows:
Firstly, by focusing on a limited ROI, the computational load of MCMC sampling is significantly reduced, and the computational efficiency is improved  by approximately 97\%. 
Secondly, incorporating statistical priors into the iterative process effectively avoids attenuating the features of micro-lesions smoothly during the iterative process in the traditional methods,
and experimentally, the relative error of the ROI is reduced by approximately 17.79\% (improving the sharpness of the corresponding boundaries).

In this paper, we propose the NWATV-Gaussian (NWATG) prior to guide the MCMC sampling within the ROI. The effectiveness of this prior is validated using three datasets: the Shepp-Logan phantom, a walnut projection provided by Finnish Inverse Problems
Society (http://fips.fi/dataset.php), and  clinical lung CT image  provided by The Cancer
Imaging Archive (TCIA: https://www.cancerimagingarchive.net/collection/lungct-diagnosis/). 
Furthermore, the same datasets are employed to test the performance of our hybrid model for global image reconstruction. Compared with the traditional iterative method (based on the box-constrained NWATV regularization), our method improves the reconstruction quality of the images.

The rest of the paper is organized as follows. In Section 2, we briefly introduce the Bayesian-Variational hybrid model, including the MCMC sampling procedure for ROI and the global image reconstruction framework. In Section 3, we validate the performance of the proposed hybrid method using the Shepp-Logan phantom, the actual walnut CT projection, and the clinical lung CT image. In Section 4, we discuss parameter selection within the statistical framework via a hierarchical modeling strategy.  In Section 5, we present conclusions and suggest directions for future research.

\section{The Bayesian-Variational hybrid model}\label{sec:intro_HM}
Traditional iterative minimization method suffers from slow convergence and  blurring of small abnormal tissues during the energy functional minimization, while statistical calculation method suffers from the low computational efficiency. To address these limitations, we propose a Bayesian-Variational hybrid model. This model introduces regional MCMC statistical sampling into the iterative optimization framework to form a joint inversion method with prior knowledge. The steps of the proposed model are as follows:
\begin{itemize}
	\item[S1.] Lesion segmentation: Segment the region of interest (ROI) of the suspected lesion;
	\item[S2.] Projection data extraction: Extract the sparse projection subset corresponding to the ROI;
	\item[S3.] Bayesian sampling of ROI: Construct a Bayesian posterior probability framework and perform MCMC sampling within the ROI;
	\item[S4.] Global image reconstruction: Incorporate the conditional mean of the MCMC samples as a prior into the iterative optimization  framework and minimize the energy functional.
\end{itemize}

In S1, the ROI segmentation is performed based on the suspected lesion areas selected by doctors through experiential assessment. The segmented binary image is $\alpha=\chi_{\mathrm{ROI}}$, where $\chi_{\mathrm{ROI}}$ is the characteristic function of ROI. The segmented ROI is shown in Fig. \ref{figure:ground}.

In S2, we apply the CT forward projection operator $A$ to the reference image $\tilde{\mathbf{u}}$ and the mask $\alpha$, and we obtain the approximate projection data $\mathbf{b}=A\mathrm{vect}\left(\tilde{\mathbf{u}}\odot(1-\alpha)\right)$ of $\Omega\backslash\mathrm{ROI}$, where $\tilde{\mathbf{u}}\odot(1-\alpha)$ represents the multiplication of the corresponding elements of two matrices, $\mathrm{vect}(\cdot)$ represents the column vectorization operator of matrix '$\cdot$', and $\Omega$ represents the imaging object. Then, $\mathbf{y}_{\mathrm{ROI}}:=\left\{\mathbf{y}-\mathbf{b}:|\mathbf{y}-\mathbf{b}|>thresholder\right\}$ is the projection generated for ROI, where $\mathbf{y}$ is the measured projection data, $thresholder$ is a given threshold.

In S3, we propose a Bayesian posterior probability framework with a nonlinear weighted anisotropic TV-Gaussian (NWATG) prior, which can effectively perform MCMC sampling. The relevant details of this Bayesian model are presented in Subsection \ref{S3}.
\subsection{The nonlinear weighted anisotropic TV-Gaussian Bayesian model}\label{S3}
In this section, we construct the statistical sampling model of ROI,\begin{equation}
	\frac{d\nu^y_n}{d\nu_0}\varpropto\exp\left\{-\frac{1}{2}\|\mathbf{y}_{\mathrm{ROI}}-A_{\mathrm{ROI}}\mathbf{u}_{\mathrm{ROI}}\|_{\Sigma_0}^2-\lambda_{\mathrm{ROI}}\|\mathbf{p}_{\mathrm{ROI}}\cdot\mathbf{\mathcal{D}}\mathbf{u}_{\mathrm{ROI}}\|_{\ell_{1}}\right\},\label{J}
\end{equation}where $\mathbf{p}_{\mathrm{ROI}}=\left(\frac{1}{|\mathbf{\mathcal{D}}_x\mathbf{u}_{\mathrm{ROI}}|^2+\beta}\frac{1}{|\mathbf{\mathcal{D}}_y\mathbf{u}_{\mathrm{ROI}}|^2+\beta}\right)$, $\Sigma_0$ is the covariance matrix of noise, $\|f\|_{\Sigma_0}^2=<\Sigma_0^{-\frac{1}{2}}f,\Sigma_0^{-\frac{1}{2}}f>$. In addition, $A_{\mathrm{ROI}}$ is the submatrix of $A$ that contains the rays passing through the ROI, $\mathbf{u}_{\mathrm{ROI}}$ is the reconstructed ROI image in column-vector form. 
$\nu_0=\mathcal{N}(0,\mathbf{C})$, and $\mathbf{C}$ represents the covariance matrix. The matrix $\mathbf{C}=(C_{i,j})$ is defined as\cite{Lv2020}
\begin{equation}
	C_{i,j}=\exp\left\{-\frac{(\mathbf{u}_{ref}[i]-\mathbf{u}_{ref}[j])^{2}}{h^{2}}\right\}.\label{C0}
\end{equation}
The term $\mathbf{u}_{ref}$ represents the reference image of the ROI, defined as the non-zero elements of $\tilde{\mathbf{u}}\odot\alpha$. Here $\mathbf{u}_{ref}[\cdot]$ represents the element of $\mathbf{u}_{ref}$. Moreover, the parameter $h$, which controls the bandwidth of the Gaussian prior, is usually set to twice the standard deviation of the noise in the reference image.

We use the preconditioned Crank-Nicolson (pCN)\cite{pCN} algorithm to sample the attenuation coefficients within the ROI, with the conditional mean of the resulting samples is denoted as $\mathbf{X}$.  Donate $J(\mathbf{u}_{\mathrm{ROI}})=\frac{1}{2}\|\mathbf{y}_{\mathrm{ROI}}-A_{\mathrm{ROI}}\mathbf{u}_{\mathrm{ROI}}\|_{\Sigma_0}^2+\lambda_{\mathrm{ROI}}\|\mathbf{p}_{\mathrm{ROI}}\cdot\mathbf{\mathcal{D}}\mathbf{u}_{\mathrm{ROI}}\|_{\ell_{1}}$, and the acceptance rate $a(\mathbf{u}_{\mathrm{ROI}},\mathbf{v})$ is defined as$$a(\mathbf{u}_{\mathrm{ROI}},\mathbf{v})=\min\{1,\exp\{J(\mathbf{u}_{\mathrm{ROI}})-J(\mathbf{v})\}\}.$$ The specific process is summarized in Algorithm \ref{algorithm:CM}.
\begin{algorithm}[h]
	\caption{Statistical sampling of ROI attenuation coefficients via pCN}
	\label{algorithm:CM}
	\begin{algorithmic}[1]
		\REQUIRE {Projection matrix $A_{\mathrm{ROI}}$,  observed data $\mathbf{y}_{\mathrm{ROI}}$, covariance matrix $\mathbf{C}$.\\ \textbf{Parameters:} $\lambda_{\mathrm{ROI}},\beta,\gamma\in\mathbb{R}^{+}$, the maximum iteration number $k_{sta}\in\mathbb{Z}^{+}$.}
		
		\ENSURE {The conditional mean $\mathbf{X}$}
		\STATE {Initialize: $\mathbf{u}_{\mathrm{ROI}}^{(0)}$}
		\FOR{k=1:$k_{sta}$}
		\STATE {Update $\mathbf{v}=\sqrt{1-\gamma^2}\mathbf{u}_{\mathrm{ROI}}^{(k-1)}+\alpha\mathbf{w},$ where $\mathbf{w}\sim\mathcal{N}(0,\mathbf{C})$.}
		\STATE {Draw $t\sim\mathcal{U}[0,1]$.}
		\IF{$t\le a(\mathbf{u}_{\mathrm{ROI}}^{(k-1)},\mathbf{v})$} \STATE{$\mathbf{u}_{\mathrm{ROI}}^{(k)}=\mathbf{v}$;}
		\ELSE 
		\STATE{$\mathbf{u}_{\mathrm{ROI}}^{(k)}=\mathbf{u}_{\mathrm{ROI}}^{(k-1)}$;}\ENDIF
		\ENDFOR
		\STATE{Calculate the conditional mean $\mathbf{X}$.}
	\end{algorithmic}
\end{algorithm}

\subsubsection{The structure of NWATG priors}\label{subsubsec:NWATG}
We consider a bounded domain $\Omega\subset\mathbb{R}^2$, with area denoted by $S_{\Omega}$. Let $\mu:\Omega\to\mathcal{X}$ denote the unknown attenuation coefficient, belonging to the function space $$\mathcal{X}=\mathcal{W}_{\eta}^{1,\infty}=\{\mu\in\mathcal{W}^{1,\infty}\vert\|\nabla\mu\|_{L^{\infty}}<\eta\}.$$
Moreover, the measured data $y$ belongs to a Hilbert space $\mathcal{Y}$ and satisfies the forward model$$y=\F(\mu)+e,$$where $\F:\mathcal{X}\to\mathcal{Y}$ represents the forward operator, and $e$ is additive Gaussian noise with mean 0 and covariance matrix $\Sigma_0$. In the Bayesian framework, the prior and posterior distributions of $\mu$ are denoted by $\nu_{pr}$, $\nu^y$, respectively. Due to the Radon-Nikodym(R-N) derivative (\cite{Stuart2010} Theorem 6.2), we can obtain:
\begin{equation}
	\frac{d\nu^y}{d\nu_{pr}}\propto\exp\{-\Phi(\mu;y)\}. \label{likehood}
\end{equation}Where $\Phi(\mu;y)=\frac{1}{2}\|\F(\mu)-y\|_{\mathcal{Y}}^2$ represents the data fidelity item. In the Bayesian framework, the posterior distribution depends on the properties of the image. Therefore, the selection of the prior distribution directly influences the posterior distribution. The hybrid prior is \begin{equation}
	\frac{d\nu_{pr}}{d\nu_0}\propto\exp\{-\R(\mu)\}.\label{pr}
\end{equation}Where $\R(\mu)=\lambda\|p\nabla\mu\|_{L^1(\Omega)}$,  which is proposed in \cite{self}. Specifically, $p=(p_1,p_2)=\left(\frac{1}{|\partial_1\mu|^2+\beta},\frac{1}{|\partial_2\mu|^2+\beta}\right)$, $\nabla=(\partial_1,\partial_2)$, and $\partial_1$, $\partial_2$ represent the first-order differential operators in the $x_1$ and $x_2$ directions, respectively. Let $\nu_0$ denote the Gaussian measure with  mean 0 and covariance matrix $\mathcal{C}$ for the reference image. Then, since \eqref{likehood} and \eqref{pr}, the posterior measure can be written as \begin{equation}
	\frac{d\nu^y}{d\nu_0}=\frac{1}{Z}\exp\{-\Phi(\mu;y)-\R(\mu)\}.\label{post}
\end{equation}  Where $Z=\int_{\mathcal{X}}\exp\{-\Phi(\mu;y)-\R(\mu)\}\nu_0(d\mu)$ is the normalized constant, and $\nu_0(\mathcal{X})=1$. 
\subsubsection{Theoretical properties of the NWATG prior}\label{subsubsec:proof}
In this section, we  demonstrate that the NWATG prior yields a well-behaved posterior distribution in $\mathcal{X}$, where the detailed proofs can be found in \cite{Stuart2010} and \cite{Yao2016}.

In fact, similar to \cite{Yao2016}, we propose the following lemma, which shows the prior is well-defined.

\begin{lemma}\label{Lemma1}
	
	The regularization functional $\R:\mathcal{X}\to\mathbb{R}$ satisfies the following properties:\\
	(1) For any $\mu\in\mathcal{X}$, $\R(\mu)$ is bounded from below, and without loss of generality we can simply assume $\R(\mu)\ge 0$.\\
	(2) For any $\delta>0$, there exists $M=M(\delta)>0$ such that $\R(\mu)\le M$ holds for all $\mu\in\mathcal{X}$ satisfying $\|\mu\|_{\mathcal{X}}\le\delta$.\\
	(3) For any $\delta>0$, there exists $L=L(\delta)>0$, such that  $$|\R(\mu_1)-\R(\mu_2)|\le L\|\mu_1-\mu_2\|_{\mathcal{X}}$$for all $\mu_1,\mu_2\in\mathcal{X}$ with $\max\{\|\mu_1\|_{\mathcal{X}},\|\mu_2\|_{\mathcal{X}}\}\le\delta$.
\end{lemma}
\begin{proof}
	Obviously, (1) is true. Next prove (2). $$\R(\mu)^2=\left(\lambda\int_{\Omega}(|p_{1}\partial_{1}\mu|+|p_{2}\partial_{2}\mu|)dx_1dx_2\right)^2.$$By fundamental inequality, we can obtain$$\R(\mu)^2\le2\lambda^2\left(\left(\int_{\Omega}|p_1\partial_{1}\mu|dx_1dx_2\right)^2+\left(\int_{\Omega}|p_2\partial_{2}\mu|dx_1dx_2\right)^2\right).$$ By Cauchy's integral inequality, we obtain$$\R(\mu)^2\le2\lambda^2S_{\Omega}\left(\int_{\Omega}|p_1|^2\cdot|\partial_{1}\mu|^2dx_1dx_2+\int_{\Omega}|p_2|^2\cdot|\partial_{2}\mu|^2dx_1dx_2\right).$$Obviously, $0< p_1\le\frac{1}{\beta},0< p_2\le\frac{1}{\beta}$, and we can get$$\begin{aligned}
		\R(\mu)^2\le&\frac{2\lambda^2S_{\Omega} }{\beta^2}\int_{\Omega}\left(|\partial_{1}\mu|^2+|\partial_{2}\mu|^2\right)dx_1dx_2\\
		\le&\frac{2\lambda^2S_{\Omega} }{\beta^2}\left(\int_{\Omega}\left(|\partial_{1}\mu|^2+|\partial_{1}\mu|^2\right)dx_1dx_2+\int_{\Omega}\mu^2dx_1dx_2\right)\\
		=&\frac{2\lambda^2S_{\Omega}}{\beta^2}\|\mu\|_{\mathcal{X}}^2.
	\end{aligned}$$Then, we can obtain the following$$\R(\mu)\le\frac{\lambda\sqrt{2S_{\Omega}}}{\beta}\|\mu\|_{\mathcal{X}}.$$For given $\delta>0$ and $M=\frac{\lambda\sqrt{2S_{\Omega}}\delta}{\beta}$, if $\|\mu\|_{\mathcal{X}}\le\delta$, we have$$\R(\mu)\le M.$$
	
	Next proves (3). $$\begin{aligned}
		&|\R(\mu_1)-\R(\mu_2)|^2\\
		=&\lambda^2\left\vert\int_{\Omega}(|p_1(\mu_1)\partial_{1}\mu_1|+|p_2(\mu_1)\partial_{2}\mu_1|)dx_1dx_2\right.\\
		-&\left.\int_{\Omega}(|p_1(\mu_2)\partial_{1}\mu_2|+|p_2(\mu_2)\partial_{2}\mu_2|)dx_1dx_2\right\vert^2\\
		=&\lambda^2\left|\int_{\Omega}(|p_1(\mu_1)\partial_{1}\mu_1|-|p_1(\mu_2)\partial_1\mu_2|)dx_1dx_2\right.\\
		+&\left.\int_{\Omega}(|p_2(\mu_1)\partial_{2}\mu_1|-|p_2(\mu_2)\partial_2\mu_2|)dx_1dx_2\right|^2
	\end{aligned}$$By fundamental inequality, we can obtain$$\begin{aligned}
		&|\R(\mu_1)-\R(\mu_2)|^2\\
		\le& 2\lambda^2\left(\int_{\Omega}(|p_1(\mu_1)\partial_{1}\mu_1|-|p_1(\mu_2)\partial_{1}\mu_2|)dx_1dx_2\right)^2\\
		+&2\lambda^2\left(\int_{\Omega}(|p_2(\mu_1)\partial_{2}\mu_1|-|p_2(\mu_2)\partial_{2}\mu_2|)dx_1dx_2\right)^2.
	\end{aligned}$$By Cauchy's integral inequality, we obtain$$\begin{aligned}
		&|\R(\mu_1)-\R(\mu_2)|^2\\
		\le&2\lambda^2S_{\Omega}\int_{\Omega}(|p_1(\mu_1)\partial_{1}\mu_1|-|p_1(\mu_2)\partial_{1}\mu_2|)^2dx_1dx_2\\
		+&2\lambda^2S_{\Omega}\int_{\Omega}(|p_2(\mu_1)\partial_{2}\mu_1|-|p_2(\mu_2)\partial_{2}\mu_2|)^2dx_1dx_2.
	\end{aligned}$$By triangle inequality, $\big||a|-|b|\big|\le|a-b|$, and we obtain $$\begin{aligned}
		&|\R(\mu_1)-\R(\mu_2)|^2\\
		\le&2\lambda^2S_{\Omega}\int_{\Omega}(p_1(\mu_1)\partial_{1}\mu_1-p_1(\mu_2)\partial_{1}\mu_2)^2dx_1dx_2\\
		+&2\lambda^2S_{\Omega}\int_{\Omega}(p_2(\mu_1)\partial_{2}\mu_1-p_2(\mu_2)\partial_{2}\mu_2)^2dx_1dx_2,
	\end{aligned}$$Therefore, 	$$\begin{aligned}
		&|\R(\mu_1)-\R(\mu_2)|^2\\
		\le&2\lambda^2S_{\Omega}\int_{\Omega}(p_1(\mu_1)\partial_{1}\mu_1-p_1(\mu_1)\partial_{1}\mu_2+p_1(\mu_1)\partial_{1}\mu_2-p_1(\mu_2)\partial_{1}\mu_2)^2dx_1dx_2\\
		+&2\lambda^2S_{\Omega}\int_{\Omega}(p_2(\mu_1)\partial_{2}\mu_1-p_2(\mu_1)\partial_{2}\mu_2+p_2(\mu_1)\partial_{2}\mu_2-p_2(\mu_2)\partial_{2}\mu_2)^2dx_1dx_2.
	\end{aligned}$$Due to the fundamental inequality, we can obtain	$$\begin{aligned}
		&|\R(\mu_1)-\R(\mu_2)|^2\\
		\le&4\lambda^2S_{\Omega}\int_{\Omega}\int_{\Omega}\left(p_1(\mu_1)^2(\partial_{1}\mu_1-\partial_{1}\mu_2)^2+(p_1(\mu_1)-p_1(\mu_2))^2\partial_{1}\mu_2^2\right)dx_1dx_2\\
		+&4\lambda^2S_{\Omega}\int_{\Omega}\int_{\Omega}\left(p_2(\mu_1)^2(\partial_{2}\mu_1-\partial_{2}\mu_2)^2+(p_2(\mu_1)-p_2(\mu_2))^2\partial_{2}\mu_2^2\right)dx_1dx_2.
	\end{aligned}$$Moreover, $$\begin{aligned}
		p_1(\mu_1)-p_1(\mu_2)=&\frac{1}{|\partial_1\mu_1|^2+\beta}-\frac{1}{|\partial_{1}\mu_2|^2+\beta}\\
		=&\frac{(\partial_{1}\mu_2)^2-(\partial_{1}\mu_1)^2}{((\partial_{1}\mu_1)^2+\beta)(\partial_{1}\mu_2)^2+\beta}\\
		\le&\frac{\partial_{1}(\mu_1+\mu_2)\partial_{1}(\mu_2-\mu_1)}{\beta^2}.
	\end{aligned}$$In the same way, we can get $$p_2(\mu_1)-p_2(\mu_2)\le\frac{\partial_{2}(\mu_1+\mu_2)\partial_{2}(\mu_2-\mu_1)}{\beta^2}.$$It is easy to show$$\begin{aligned}
		&|\R(\mu_1)-\R(\mu_2)|^2\\
		\le&4\lambda^2S_{\Omega}\int_{\Omega}\left(\frac{1}{\beta^2}+\frac{\partial_{1}(\mu_1+\mu_2)^2\partial_{1}\mu_2^2}{\beta^4}\right)\partial_{1}(\mu_1-\mu_2)^2dx_1dx_2\\
		+&4\lambda^2S_{\Omega}\int_{\Omega}\left(\frac{1}{\beta^2}+\frac{\partial_{2}(\mu_1+\mu_2)^2\partial_{2}\mu_2^2}{\beta^4}\right)\partial_{2}(\mu_1-\mu_2)^2dx_1dx_2.
	\end{aligned}$$For any $\mu_1,\mu_2\in\mathcal{X}$, there exists $\eta>0$, such that $|\partial_i\mu_j|<\eta,i,j=1,2.$ Obviously, $|\partial_i(\mu_1+\mu_2)|\le|\partial_i\mu_1|+|\partial_i\mu_2|<2\eta,i=1,2.$ It follows immediately that\begin{equation}
		\begin{aligned}
			&|\R(\mu_1)-\R(\mu_2)|^2\\
			\le&\frac{4\lambda^2(4\eta^4+\beta^2)}{\beta^4}S_{\Omega}\left(\int_{\Omega}\partial_{1}(\mu_1-\mu_2)^2+\partial_{2}(\mu_1-\mu_2)^2\right)dx_1dx_2\\
			\le&\frac{4\lambda^2(4\eta^4+\beta^2)}{\beta^4}S_{\Omega}\left(\int_{\Omega}\partial_{1}(\mu_1-\mu_2)^2+\partial_{2}(\mu_1-\mu_2)^2+(\mu_1-\mu_2)^2\right)dx_1dx_2\\
			=&\frac{4\lambda^2(4\eta^4+\beta^2)}{\beta^4}S_{\Omega}\|\mu_1-\mu_2\|_{\mathcal{X}}^2.
		\end{aligned}
	\end{equation}Let $L=\frac{4\lambda^2(4\eta^4+\beta^2)}{\beta^4}S_{\Omega}$, we obtain $$|\R(\mu_1)-\R(\mu_2)|\le L\|\mu_1-\mu_2\|_{\mathcal{X}}.$$We need to use fact that for any $a\ge0,b\ge0,|\exp(-a)-\exp(-b)|\le\min\{1,|a-b|\}$. Then, the prior measure is well-defined, with \begin{equation}
		|\exp(-\R(\mu_1))-\exp(-\R(\mu_2))|\le L\|\mu_1-\mu_2\|_{\mathcal{X}}.\label{eR_lip}
	\end{equation}Moreover, $$Z=\int_{\mathcal{X}}\exp(-\R(\mu))d\nu_0(\mu)\le\nu_0(\mathcal{X})=1<\infty.$$The prior measure is normalizable, and the conclusion holds under the stated conditions.
\end{proof}

Actually, with the Assumption \ref{assum_Phi}, it can be proved that the posterior measure is well-defined.

\begin{assumption}\label{assum_Phi}
	
	Assume $\Phi:\mathcal{X}\times\mathcal{Y}\to\mathbb{R}$ satisfies the following properties:\\
	(1) For every $\mu\in\mathcal{X}$ and every $y\in\mathcal{Y}$, $\Phi(\mu;y)$ is bounded below, and without loss of generality we can simply assume $\Phi(\mu;y)\ge 0$.\\
	(2) For any $r>0$, there is an $k(r)>0$, such that for all $\mu\in\mathcal{X},\ y\in\mathcal{Y}$ with $\max\{\|\mu\|_{\mathcal{X}},\|y\|_{\mathcal{Y}}\}<r$, $$\Phi(\mu;y)\le k.$$\\
	(3) For any $r>0$, there exists an $L(r)>0$, such that for all $\mu_1,\mu_2\in\mathcal{X},y\in\mathcal{Y}$ with $\max\{\|\mu_1\|_{\mathcal{X}},\ \|\mu_2\|_{\mathcal{X}},\|y\|_{\mathcal{Y}}\}<r$, $$|\Phi(\mu_1;y)-\Phi(\mu_2;y)|\le L\|\mu_1-\mu_2\|_{\mathcal{X}}.$$\\
	(4) For any $\epsilon>0,r>0$, there exists $C=C(\epsilon,r)\in\mathbb{R}$, such that for all $y_1,\ y_2\in\mathcal{Y}$ with $\max\{\|y_1\|_{\mathcal{Y}},\|y_2\|_{\mathcal{Y}}\}<r$, for all $\mu\in\mathcal{X}$, $$|\Phi(\mu;y_1)-\Phi(\mu;y_2)|\le\exp(\epsilon\|\mu\|_{\mathcal{X}}^2+C)\|y_1-y_2\|_{\mathcal{Y}}.$$
\end{assumption}In fact, the data fidelity term $\Phi$ satisfies Assumption \ref{assum_Phi}. This is essentially equivalent to $\Phi+\R$ satisfies Assumption 2.6 in \cite{Stuart2010}, only slightly different in the first one.

Then, the following theorem can be proved. \begin{theorem}\label{post_well}
	Assume $\Phi(\mu;y)$ satisfies Assumption \ref{assum_Phi}, and $\nu_0$ is a Gaussian measure with $\nu_0(\mathcal{X})=1$. Then we have the following:
	
	(1)  The posterior measure $\nu^y$ is well-defined.
	
	(2) $\nu^y$ is Lipschitz continuous in the data $y$ with respect to the Hellinger distance. Specifically, for any $y\in\mathcal{Y}$, $\nu^y\ll\nu_0$ with Radon-Nikodym derivative given by \eqref{likehood},  $\nu^y,\nu^{y^{'}}$ represent the measures of $y$ and $y^{'}$, respectively. Then for every $r>0$, there exists $M=M(r)>0$, such that for all $y,y^{'}\in\mathcal{Y}$ with $\max\{\|y\|_{\mathcal{Y}},\|y^{'}\|_{\mathcal{Y}}\}<r$, the $\mathrm{Hellinger}$ distance\cite{Kuo1975} satisfies  $$d_{Hell}(\nu^y,\nu^{y^{'}})\le M\|y-y^{'}\|_{\mathcal{Y}}.$$
\end{theorem}The conclusion coincides with Theorem 4.2 in \cite{Stuart2010}; therefore, the proof is omitted here.

In fact, a finite-dimensional approximation of $\mu$ is employed. Let $\mathbf{u}$ represent the $n$-dimensional approximation of $\mu$; then the forward problem can be written as $$\mathbf{y}=A\mathbf{u}+\mathbf{e},$$where $A$ is the discrete matrix of the forward  operator $\mathcal{F}$, and $\mathbf{y},\mathbf{e}\in\mathbb{R}^m$ represent the $m$-dimensional approximations of $y,e$, respectively. The convergence of the posterior measure for $\mathbf{u}$ is established in Theorem \ref{finit_1}.
\begin{theorem}\label{finit_1}
	Let $\{e_k\}_{k=1}^{\infty}$ represent a complete orthonormal basis of $\mathcal{X}$. The $n$-dimensional approximation of $\mu$ is given by $$\mathbf{u}=\sum_{k=1}^n\langle\mu,e_k\rangle e_k.$$
	The corresponding posterior measure $\nu_n^y$ for $\mathbf{u}$ then admits the Radon-Nikodym derivative
	\begin{equation}
		\frac{d\nu^y_n}{d\nu_0}=\frac{1}{Z_n}\exp\{-\Phi(\mathbf{u};\mathbf{y})-\mathrm{R}(\mathbf{u})\},\label{post_n}
	\end{equation}where $Z_n$ is the normalization constant. If $\Phi$ satisfies the Assumption \ref{assum_Phi} and $\mathrm{R}$ satisfies the Lemma \ref{Lemma1}, then the Hellinger distance between the posterior $\nu^y$ and its finite-dimensional approximation $\nu_n^y$ converges to zero, i.e., $$d_{Hell}(\nu^y,\nu^y_n)\to 0,\qquad \text{as}\ n\to\infty.$$
\end{theorem} Furthermore, we introduce the discretized measure $\nu_{N_1,N_2}^y$ defined via $$	\frac{d\nu^y_{N_1,N_2}}{d\nu_0}\propto\exp(-\Phi_{N_1}(\mu;y)-\R_{N_2}(\mu)),$$where $\Phi_{N_1}(\mu;y)$ and $\R_{N_2}$ denote $N_1$- and $N_2$-dimensional approximations of the data fidelity term $\Phi$ and the regularization term $\R$, respectively. The convergence of measure $\nu_{N_1,N_2}^y$ is established in \cite{Yao2016}. Consequently, Theorem \ref{finit_1} establishes the applicability of the NWATG model to image reconstruction in arbitrary dimensions. 
\begin{remark} 
	For any $\mu\in\mathcal{X}$, the bound  $|\partial_i\mu|<\eta$ holds for $i=1,2$. To bound  $\|\mu\|_{\mathcal{X}}\le\delta$, it suffices to estimate the following inequality:
	$$\int_{\Omega}(|\partial_{1}\mu|^2+|\partial_{2}\mu|^2)dx_1dx_2\le\int_{\Omega}(|\partial_{1}\mu|^2+|\partial_{2}\mu|^2+|\mu|^2)dx_1dx_2=\|\mu\|_{\mathcal{X}}^2 
	. $$ Since $|\partial_{i}\mu|<\eta$, we have
	$$\int_{\Omega}(|\partial_{1}\mu|^2+|\partial_{2}\mu|^2)dx<\int_{\Omega}2\eta^2dx=2\eta^2 S_\Omega.$$
	Thus, a sufficient condition for $\|\mu\|_{\mathcal{X}}\le \delta$ is that $2\eta^2 S_{\Omega} \le \delta^2$, or equivalently, $$\eta^2\le\frac{\delta^2}{2S_{\Omega}}.$$
\end{remark}

In the following section, we detail step S4, which constructs a hybrid model with integrated prior information by incorporating local statistical sampling results into the energy functional minimization process.

\subsection{Reconstruction method incorporating prior knowledge}\label{subsec:S4}

In S4,  the CM estimate from the sampling results is incorporated as prior knowledge into the energy functional minimization to guide the CT image reconstruction. To prevent the smooth attenuation of subtle lesion features commonly encountered in traditional iterative methods,  we integrate the MCMC sampling within the ROI into the reconstruction workflow. This leads to a hybrid reconstruction paradigm that minimizes an energy functional comprising a data fidelity term and boundary-preserving regularization terms.  Specifically, the CM estimate $\mathbf{X}=(X_{i,j})$ from the MCMC sampling in S3  is used as the prior knowledge to construct the regularization term, which is included in the energy functional for minimization. This process includes four stages: construction of local regularizer, construction of composite regularization terms,  global minimization of the energy functional, and imposition of a box constraint. We use the shorthand $\R_1=I_\tau$ and  $\R_2=\mathrm{ROI}\setminus\R_1$, and the segmentation of ROI is shown in Fig. \ref{R12}. 

\begin{figure}[htb]
	\centering
	\includegraphics[width=4cm]{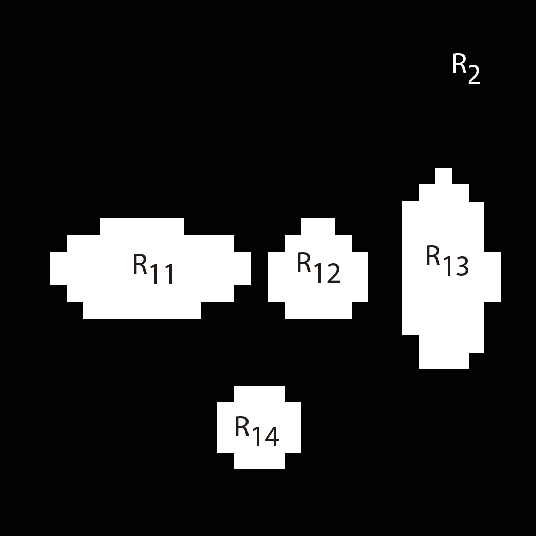}
	\caption{The segmentation of ROI. We take the segmentation of the Shepp-Logan image ROI as an example. $\R_1=\R_{11}\cup \R_{12}\cup \R_{13}\cup \R_{14}$.}\label{R12}
\end{figure}
\begin{itemize}
	\item[Step1.]\ \ \  Construction of local regularizer: According to the index set $I_\tau$ of ROI, we define
	\begin{equation*}
		\mathbf{X}_0=\left\{\begin{aligned}
			&X_{i,j},\quad&(i,j)\in I_\tau,\\
			&0,&(i,j)\in\Omega\setminus I_\tau.
		\end{aligned}\right.
	\end{equation*}Here, $\mathbf{X}_0$ is expressed in vector form.
	\item[Step2.]\ \ \  Construction of composite regularization terms: \par
	(1) Piecewise-constant constraint: Since the image $\mathbf{u}$ is piecewise constant within the ROI, we add the regularization terms $\|\mathbf{\mathcal{D}}\mathbf{u}\|_{\ell_{2}(\R_1)}^2$and $\|\mathbf{\mathcal{D}}\mathbf{u}\|_{\ell_{2}(\R_2)}^2$to the energy functional; \par
	(2) Edge-alignment constraint: To ensure that $\mathbf{u}$ preserves the boundary of the local regularizer $\mathbf{X}_0$ inside the ROI, we add the regularization term $\|\mathbf{\mathcal{D}}(\mathbf{u}-\mathbf{X}_0)\|_{\ell_{2}(\mathrm{ROI})}^2$ to the energy functional;\par
	\item[Step3.]\ \ \  Global minimization of the energy functional:\par
	(1) Formulation of the optimization problem: The reconstruction is obtained by solving \begin{equation}
		\begin{aligned}
			\mathbf{u}^{\ast}=\underset{\mathbf{u}}{\arg\min}&\left\{\frac{1}{2}\|\mathbf{y}-A\mathbf{u}\|_{\ell_{2}}^2+\lambda\|\mathbf{p}\cdot\mathbf{\mathcal{D}}\mathbf{u}\|_{\ell_{1}}+\frac{\rho_1}{2}\|\mathbf{\mathcal{D}}\mathbf{u}\|_{\ell_2(\R_1)}^2\right.\\
			&\left.+\frac{\rho_2}{2}\|\mathbf{\mathcal{D}}\mathbf{u}\|_{\ell_2(\R_2)}^2+\frac{\rho_3}{2}\|\mathbf{\mathcal{D}}(\mathbf{u}-\mathbf{X}_0)\|_{\ell_2(\mathrm{ROI})}^2\right\},
		\end{aligned}\label{g1}\end{equation}
	where $\mathbf{u}\in\mathbb{R}^n$, $\mathbf{y}\in\mathbb{R}^m$,   $\mathbf{p}=\left(\omega(\mathbf{\mathcal{D}}_{x}\mathbf{u});\omega(\mathbf{\mathcal{D}}_{y}\mathbf{u})\right)\in\mathbb{R}^{2n},\mathbf{\mathcal{D}} \mathbf{u}=(\mathbf{\mathcal{D}}_{x}\mathbf{u};\mathbf{\mathcal{D}}_{y}\mathbf{u})\in\mathbb{R}^{2n}$, and $\mathbf{\mathcal{D}}_{x}$, $\mathbf{\mathcal{D}}_{y}\in\mathbb{R}^{n\times n}$ are the first-order difference operators along the $x$
	and $y$ directions, respectively. $\omega(\cdot)=\frac{1}{|\cdot|^{2}+\beta}$ with a small $\beta>0$  to avoid division by zero. \par
	(2) Implementation of the solution method: Define index sets $M$, $M_1$, $M_2$ that take the value $1$ inside $\mathrm{ROI}$, $\R_1$, $\R_2$, respectively, and $0$ elsewhere.  Using element‑wise multiplication ($\odot$), problem \eqref{g1} is equivalent to \begin{equation}
		\begin{aligned}
			\mathbf{u}^{\ast}=\underset{\mathbf{u}}{\arg\min}&\left\{\frac{1}{2}\|\mathbf{y}-A\mathbf{u}\|_{\ell_{2}}^2+\lambda\|\mathbf{p}\cdot\mathbf{\mathcal{D}}\mathbf{u}\|_{\ell_{1}}+\frac{\rho_1}{2}\|\mathbf{\mathcal{D}}(M_1\odot\mathbf{u})\|_{\ell_2}^2\right.\\
			&\left.+\frac{\rho_2}{2}\|\mathbf{\mathcal{D}}(M_2\odot\mathbf{u})\|_{\ell_2}^2+\frac{\rho_3}{2}\|\mathbf{\mathcal{D}}(M\odot(\mathbf{u}-\mathbf{X}_0))\|_{\ell_2}^2\right\}.
		\end{aligned}\label{g2}	\end{equation} The minimization is performed with the Alternating Direction Method of Multipliers (ADMM) \cite{ADMM}.\par
	\item[Step4.]\ \ \  Imposition of a box constraint: The iterative solution is projected onto the interval $[a,b]$ via the operator $\Pi_{[a,b]}$, with $a$ and $b$ being the physical lower and upper limits of $\mu$.
\end{itemize}

To solve the Equation \eqref{g2}, we construct the augmented Lagrangian functional:
\begin{equation}
	\begin{aligned}
		\mathcal{L}(\mathbf{u},\mathbf{d},\mathbf{p},\mathbf{b})&=\frac{1}{2}\|\mathbf{y}-A\mathbf{u}\|_{\ell_{2}}^2+\lambda\|\mathbf{p}\cdot\mathbf{d}\|_{\ell_{1}}+\frac{\rho_1}{2}\|\mathbf{\mathcal{D}}(M_1\odot\mathbf{u})\|_{\ell_2}^2\\
		&+\frac{\rho_2}{2}\|\mathbf{\mathcal{D}}(M_2\odot\mathbf{u})\|_{\ell_2}^2+\frac{\rho_3}{2}\|\mathbf{\mathcal{D}}(M\odot(\mathbf{u}-\mathbf{X}_0))\|_{\ell_2}^2\\
		&+<\mathbf{\mathcal{D}}\mathbf{u}-\mathbf{d},\mathbf{b}>
		+\frac{\rho}{2}\|\mathbf{\mathcal{D}}\mathbf{u}-\mathbf{d}\|_{\ell_{2}}^2,
	\end{aligned}\label{aug_lagrangian}
\end{equation}where $\mathbf{d}$ is the auxiliary variable, $\mathbf{b}$ is the Lagrangian multiplier, and $\rho$ is the scalar
penalty parameter. The minimization of Equation \eqref{aug_lagrangian} is carried out via ADMM.

For ease of calculation, we treat the $n\times 1$ vectors $M$, $M_1$, and $M_2$ as diagonal entries to construct the $n\times n$ diagonal matrices, which we denote by the same symbols $M$, $M_1$, and $M_2$ for simplicity.

Specifically, given an initial guess $(\mathbf{d}^{(0)},\mathbf{p}^{(0)},\mathbf{b}^{(0)})$, the variable $\mathbf{u}$ is updated iteratively according to the following scheme:
\begin{subequations}
	\begin{align}
		\mathbf{u}^{(k+1)}&=\underset{\mathbf{u}}{\arg\min}\mathcal{L}(\mathbf{u},\mathbf{d}^{(k)},\mathbf{p}^{(k)},\mathbf{b}^{(k)});\label{u-sub}\\
		\mathbf{d}^{(k+1)}&=\underset{\mathbf{d}}{\arg\min}\mathcal{L}(\mathbf{u}^{(k+1)},\mathbf{d},\mathbf{p}^{(k)},\mathbf{b}^{(k)});\label{d-sub}\\
		\mathbf{p}^{(k+1)}&=\left(\omega(\mathbf{\mathcal{D}}_{x} \mathbf{u}^{(k+1)});\omega(\mathbf{\mathcal{D}}_{y} \mathbf{u}^{(k+1)})\right);\label{p-sub}\\
		\mathbf{b}^{(k+1)}&=\mathbf{b}^{(k)}+\rho\left(\mathbf{\mathcal{D}} \mathbf{u}^{(k+1)}-\mathbf{d}^{(k+1)}\right);\label{b-sub}
	\end{align}
\end{subequations}
For Equations \eqref{u-sub} and \eqref{d-sub}, the minimizers admit the following closed-form expressions:
\begin{equation}\label{u-box-update}
	\begin{aligned}
		\mathbf{u}^{(k+1)}	=&\left[A^{T}A+\rho\mathbf{\mathcal{D}}^{T}\mathbf{\mathcal{D}}+\rho_1M_{1}^{T}\mathbf{\mathcal{D}}^{T}\mathbf{\mathcal{D}}M_{1}+\rho_{2}M_2^{T}\mathbf{\mathcal{D}}^{T}\mathbf{\mathcal{D}}M_2+\rho_3M^{T}\mathbf{\mathcal{D}}^{T}\mathbf{\mathcal{D}}M\right]^{-1}\\
		& \left[A^{T}\mathbf{y}+\rho\mathbf{\mathcal{D}}^{T}\mathbf{d}^{(k)}-\mathbf{\mathcal{D}}^{T}\mathbf{b}^{(k)}+\rho_3M^{T}\mathbf{\mathcal{D}}^{T}\mathbf{\mathcal{D}}M\tilde{\mathbf{u}}\right],
	\end{aligned}
\end{equation}
and
\begin{equation}\label{d-update}
	\mathbf{d}^{(k+1)}[i]=h_{\frac{\lambda|\mathbf{p}^{(k)}[i]|}{\rho}}\left(\mathbf{\mathcal{D}} \mathbf{u}^{(k+1)}[i]+\frac{1}{\rho}\mathbf{b}^{(k)}[i]\right).
\end{equation}
Here, $\mathcal{I}$ is an $n\times n$ identity matrix; $\mathbf{b}^{(k)}[i]$, $\mathbf{p}^{(k)}[i]$, $\mathbf{u}^{(k+1)}[i]$, and $\mathbf{d}^{(k+1)}[i]$  are the $i$-th components of $\mathbf{b}^{(k)}$, $\mathbf{p}^{(k)}$, $\mathbf{u}^{(k+1)}$, and $\mathbf{d}^{(k+1)}$, respectively; and $h_{g}(\cdot)$ denotes the soft threshold formula, which is defined as  follows\cite{soft}:
\begin{equation*}
	h_{g}(\cdot)=
	\begin{cases}
		\cdot-g\sgn(\cdot),& \mbox{if }|\cdot|>g\\
		0,&\mbox{otherwise},
	\end{cases}
\end{equation*}
where $\sgn$ is the $\mbox{sign}$ function. 

We conclude this section by summarizing the above procedure as a reconstruction algorithm, whose pseudocode is presented in Algorithm \ref{algorithm:hybrid}.

\begin{algorithm}[h]
	\caption{Hybrid Bayesian-Variational Reconstruction via ADMM}
	\label{algorithm:hybrid}
	\begin{algorithmic}[1]
		\REQUIRE {Coefficient matrix A, projection data $\mathbf{y}$, reference image $\tilde{\mathbf{u}}$, The lower and upper bounds of the attenuation coefficient $[a,b]$.\\ \textbf{Parameters}: $\lambda,\beta,\rho,\rho_1,\rho_{2},\rho_3\in\mathbb{R}^{+}$,  a tolerance $\bar{\epsilon}$ and the maximum iteration number $k_{iter}\in\mathbb{Z}^{+}$.}
		
		\ENSURE {the reconstructed image $\mathbf{u}_{\mathrm{final}}$.}
		\STATE {Initialize: $\mathbf{d}^{(0)}=\mathbf{0},\mathbf{p}^{(0)}=(\frac{1}{\beta})\mathbf{1},\mathbf{b}^{(0)}=\mathbf{0},\mathbf{u}^{(0)}=\mathbf{0}.$}
		\FOR{k=1:kmax}
		\STATE {Update $\mathbf{u}^{(k)}$ using (\ref{u-box-update}).}
		\STATE {Update $\mathbf{d}^{(k)}$ using (\ref{d-update}).}
		\STATE {Update $\mathbf{p}^{(k)}$ using (\ref{p-sub}).}
		\STATE {Update $\mathbf{b}^{(k)}$ using (\ref{b-sub}).}
		\IF{$\|\mathbf{u}^{(k)}-\mathbf{u}^{(k-1)}\|_{\ell_{2}}<\bar{\epsilon}$} \STATE{break} \ENDIF
		\ENDFOR
		\STATE {\textbf{Result optimization:} $\mathbf{u}_{\mathrm{final}}=\min\{\max\{\mathbf{u},a\},b\}$.}
	\end{algorithmic}
\end{algorithm}

\section{Experiments}\label{sec:experiment}
In this section, we evaluate the effectiveness of the proposed  hybrid Bayesian-Variational method on three datasets: the Shepp-Logan phantom, the actual walnut CT projection provided by Finnish Inverse Problems Society (http://fips.fi/dataset.php), and the clinical lung CT images provided by The Cancer Imaging Archive  (https://www.cancerimagingarchive.net/collection/lungct-diagnosis/).   Additionally, the  walnut CT data can be found in ZENODO  (https://zenodo.org/record/1254206). A circular inclusion simulating an abnormal lesion (e.g., a tumor) is introduced into the Shepp-Logan phantom. Fig. \ref{figure:ground} displays the experimental data used in this study.
\begin{figure}[h]
	\centering
	\includegraphics[width=0.8\linewidth]{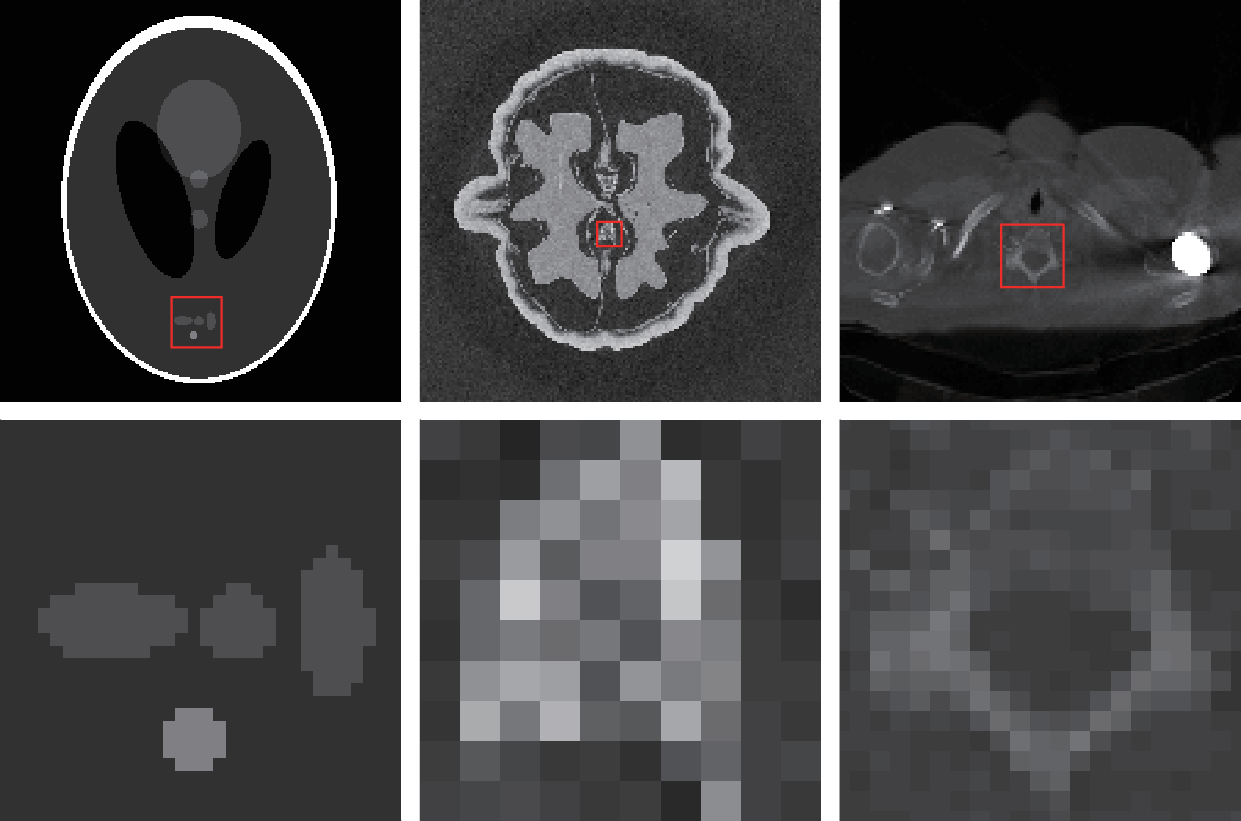}
	\caption{Experimental data and corresponding ROIs. The top row displays (from left to right) the  Shepp-Logan phantom with a simulated circular lesion, the actual walnut image and a clinical lung image, and the bottom row presents the corresponding ROIs. The pixel dimensions of the ROIs are 32$\times$32, 10$\times$10, and 20$\times$20, respectively.}\label{figure:ground}
\end{figure} The reference images, which can be regarded as historical CT scans of the same patient acquired with  the same scanner, are presented in Fig. \ref{figure:ref}. 
\begin{figure}[hbt]
	\centering
	\includegraphics[width=0.8\linewidth]{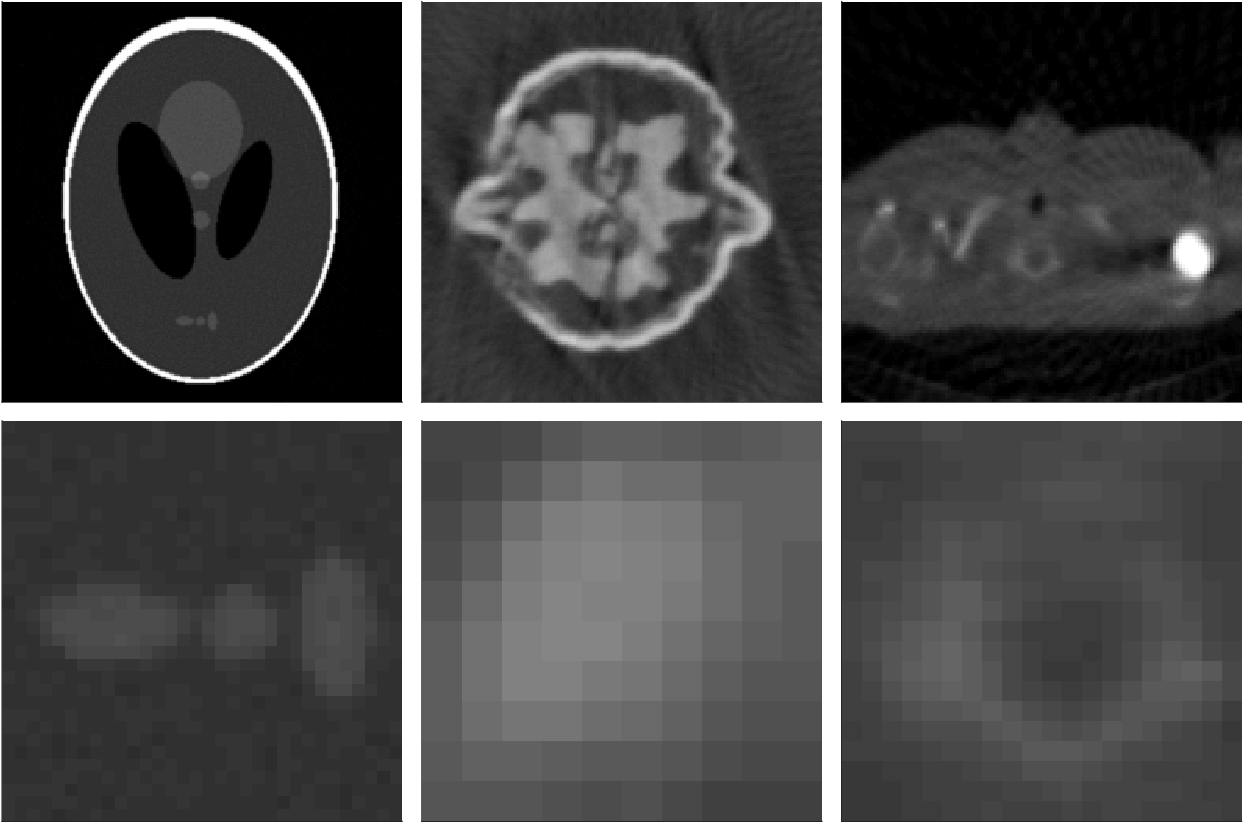}
	\caption{Reference images and corresponding ROIs. Arranged left to right are the reference images for the Shepp-Logan phantom, the walnut projection, and the clinical lung CT image, along with their corresponding ROIs. The reference images (top row) are generated using the generalized Tikhonov regularization with  parameters 200, 10 and 10 respectively.  In the Shepp-Logan case, the projection data contains 0.5\% Gaussian white noise.}\label{figure:ref}
\end{figure}
\subsection{Experiment setup}\label{subsec:setup}
To demonstrate the advantages of the proposed hybrid model, we compare its reconstruction results with those obtained using the box-constrained NWATV method\cite{self}. For quantitative evaluation, we compute the $\mathrm{L}^2$ relative  error $\RE$, the $\mathrm{H}^1$ relative error $\widetilde{\RE}$, and the mean‑square error $\MSE$, defined as follows:
$$
\RE=\frac{\|\mathbf{u}^{\ast}-\mathbf{u}_{0}\|_{\ell_{2}}}{\|\mathbf{u}_{0}\|_{\ell_{2}}},
$$
$$	\widetilde{\RE}=\frac{\sqrt{\|\mathbf{u}^{\ast}-\mathbf{u}_{0}\|_{\ell_{2}}^{2}+\|\mathbf{\mathcal{D}}(\mathbf{u}^{\ast}-\mathbf{u}_{0})\|_{\ell_{2}}^{2}}}{\sqrt{\|\mathbf{u}_{0}\|_{\ell_{2}}^{2}+\|\mathbf{\mathcal{D}}\mathbf{u}_{0}\|_{\ell_{2}}^{2}}},
$$and$$
\MSE=\frac{\|\mathbf{u}^{\ast}-\mathbf{u}_{0}\|_{\ell_{2}}^2}{n}.
$$Where $\mathbf{u}_0$ represents the ground truth, and the $\mathbf{u}^{\ast}$ represents the reconstruction result.

In addition, we compute the peak signal-to-noise ratio $\PSNR$ and the structural similarity index $\SSIM$ \cite{SSIM}, which are defined as follows:
\begin{equation*}
	\PSNR=10\log_{10}\frac{\max(\mathbf{u}^{\ast}\odot\mathbf{u}^{\ast})}{\MSE},
\end{equation*}
and
\begin{equation*}
	\SSIM=\frac{(2\mu_{\ast}\mu_0+C_{1})(2\sigma_{\ast0}+C_{2})}{(\mu_{\ast}^{2}+\mu_0^{2}+C_{1})(\sigma_{\ast}^{2}+\sigma_0^{2}+C_{2})},
\end{equation*} 
where $\odot$ denotes element‑wise multiplication.
The terms $\mu_\ast$ and $\sigma_\ast$ are the local mean and standard deviation of the reconstruction $\mathbf{u}^*$, while $\mu_0$, $\sigma_0$ are those of the reference image $\mathbf{u}_0$; $\sigma_{\ast0}$ denotes the cross-covariance between  $\mathbf{u}^{\ast}$ and $\mathbf{u}_{0}$.
The constants $C_1=10^{-4}$ and $C_2=9\times10^{-4}$ are the default values of the MATLAB built-in function ``ssim".

For the Shepp-Logan phantom, the reconstructed image size is $256 \times 256$ pixels. For the actual walnut CT data, the reconstructed image size is $164 \times 164$. For the clinical lung image, we select the first image of patient R\_172 from the dataset. Its original resolution is $512 \times 512$, and we downsample it uniformly to obtain a $128 \times 128$ ground‑truth image, from which we generate the corresponding projection data via the Radon transform.  Equation \eqref{u-box-update} is solved using the generalized minimal residual algorithm (GMRES) \cite{Saad_1986}. 

The reconstructions are carried out using Matlab 2018a (The MathWorks, Inc., Natick, MA, USA) on a workstation with 1.60 GHz Inter (R) Core (TM) i5-8250U CPU, 8.00 GB memory, Windows 10 operating system. Additionally, we use the MATLAB package AIR Tools \uppercase\expandafter{\romannumeral2} to simulate the \textit{parallel beam}  for the CT scanning \cite{AIR}.
\subsection{The reconstruction of NWATG statistical model}\label{section:statistical experiment}

In this section, firstly, we demonstrate the NWATG statistical model with the Shepp-Logan phantom.  
The reference image is generated as follows. 
First, projection data are simulated with 600 angles and then contaminated with 0.5\% Gaussian white noise. These noisy data are reconstructed via generalized Tikhonov regularization (parameter 200) to produce the reference image, which is displayed in Fig. \ref{figure:ref}. 
For the MCMC sampling, projection data are simulated at 30 uniformly spaced angles over $[0^\circ, 179^\circ]$ and corrupted with 1\% Gaussian white noise. The initial value $\mathbf{u}^{(0)}$ is obtained by applying the  generalized Tikhonov regularization with parameter 1000 to the projection data restricted to the ROI.  We collect 10,000 MCMC samples, discarding the first 8,000 as burn-in. The step‑size parameter $\alpha$ is tuned to achieve an acceptance rate of approximately 25\%. The CM estimate from the MCMC samples is shown in Fig. \ref{result_0.005ref}.  The relevant parameters and the resulting acceptance rates are summarized in Table \ref{table_cm}.
\begin{figure}[h]
	\centering
	\includegraphics[width=0.8\linewidth]{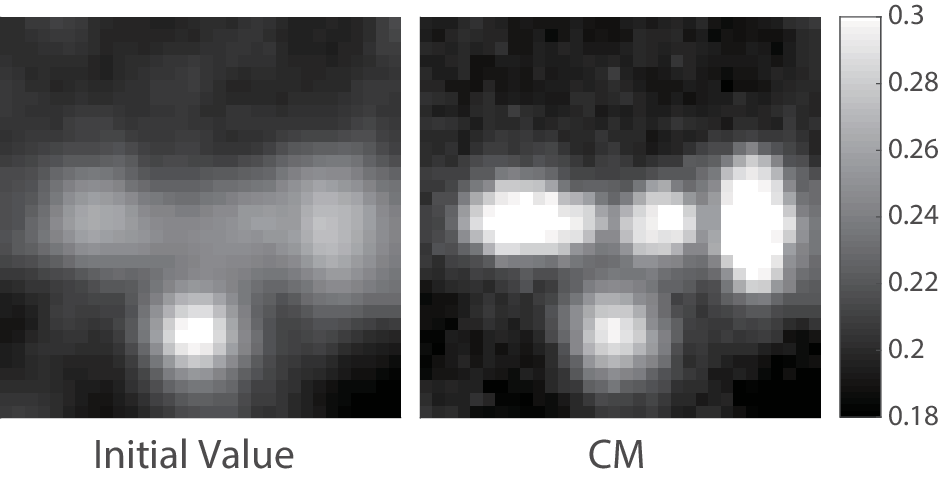}
	\caption{The CM estimate from MCMC sampling within the ROI of Shepp-Logan phantom.  The left is the initial value, and the right is the CM estimate. The image size is 32$\times$32.}\label{result_0.005ref}
\end{figure}

Finally, we discuss statistical sampling on actual CT data,  including a walnut projection and a clinical lung image.  The reference images, obtained via generalized Tikhonov regularization with a parameter of 10, are shown in Fig. \ref{figure:ref}. The initial value $\mathbf{u}^{(0)}$ for each experiment is set to the reconstruction result obtained by applying generalized Tikhonov regularization to the ROI projection data, using a parameter of 200 for the walnut and 1000 for the lung.
In the walnut experiment, we use the data from the dataset \textbf{Data164}\cite{walnut}. Then, the projection data $\mathbf{y}\in\mathbb{R}^{164\times50}$ and the matrix $A\in\mathbb{R}^{8200\times26896}$ are constructed by taking the first 50 columns of data $\bar{\mathbf{y}}\in\mathbb{R}^{164\times120}$ and the first 8200 rows of matrix $\bar{A}\in\mathbb{R}^{19680\times 26896}$, respectively.  This configuration corresponds to uniformly sampling 50 projection angles over the range of \ang{0} to \ang{149}.
In the clinical lung experiment, the first image of the patient $\R\_172$\cite{lung} is selected. The $512\times512$ image is uniformly downsampled to obtain a $128\times128$ image $\mathbf{u}_0$, as shown in Fig. \ref{figure:ground}. A parallel beam scanning geometry is simulated with 181 detectors and 30 projection angles to construct the system matrix $A$. The corresponding projection data are generated as  $\mathbf{y}=A\mathbf{u}_0$. The MCMC-sampled CM estimates for the ROI projection data are shown in Fig. \ref{result_walnut} and Fig. \ref{result_lung}. The corresponding parameters and acceptance rates are summarized in Table \ref{table_cm}.
\begin{figure}[h]
	\centering
	\includegraphics[width=0.8\linewidth]{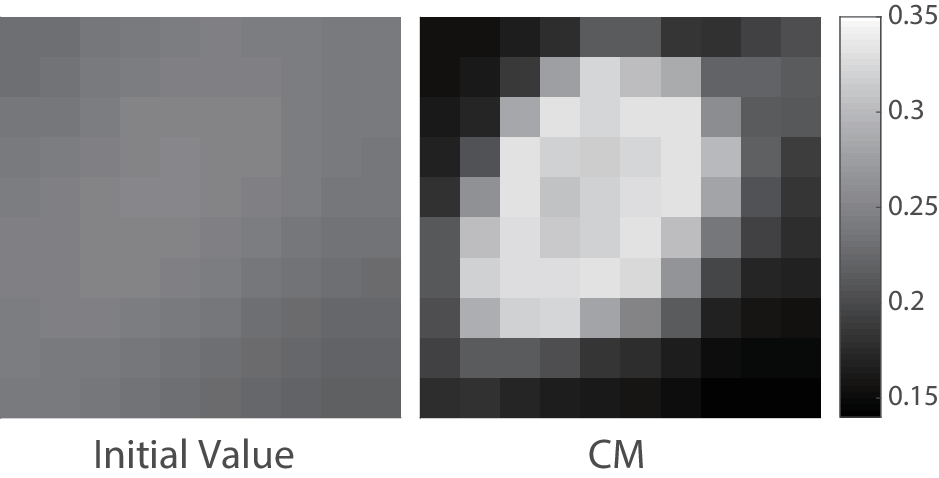}
	\caption{The CM estimate from MCMC sampling within the ROI of walnut image. The left is the initial value, and the right is the CM estimate. The image size is 10$\times$10.}\label{result_walnut}
\end{figure}
\begin{figure}[h]
	\centering
	\includegraphics[width=0.8\linewidth]{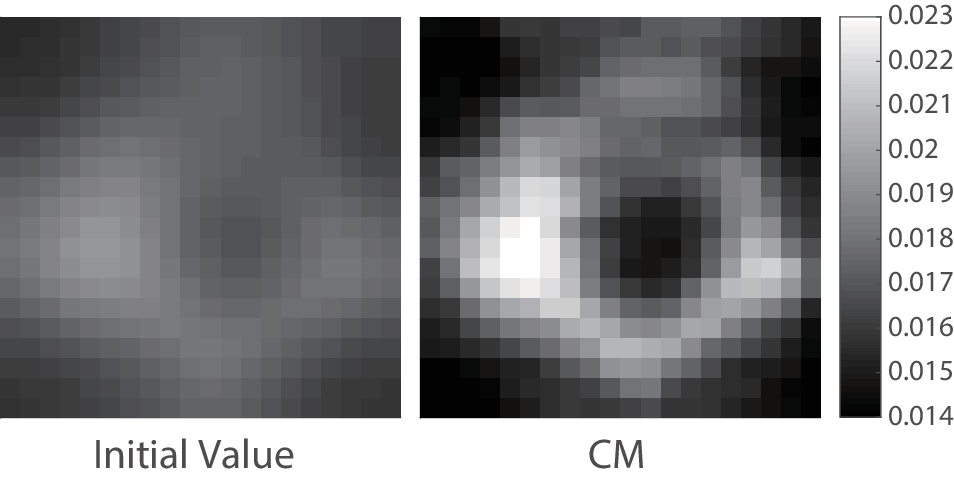}
	\caption{The CM estimate from MCMC sampling within the ROI of clinical lung image. The left is the initial value, and the right is the CM estimate. The image size is 20$\times$20.}\label{result_lung}
\end{figure}
\begin{table}
	\caption{Statistical sampling parameters and acceptance rates.}\label{table_cm}
	\begin{tabular}{l|lll}
	\hline\noalign{\smallskip}
		&$\lambda_{\mathrm{ROI}}$&$\gamma$&Acceptance rate (\%)\\
	\noalign{\smallskip}\hline\noalign{\smallskip}
		Shepp-Logan&0.001&0.004&23.35\\
		Walnut&0.001&0.008&25.15\\
		Clinical lung &0.001&6$\times10^{-5}$&26.95\\
\noalign{\smallskip}\hline
	\end{tabular}
\end{table}
For all experiments in this section, we collect 10000 MCMC samples, discarding the first 8,000 as the burn-in period. The acceptance rate $r$ is calculated from the remaining 2,000 samples using:
$$r=\frac{N_{accepted}}{2000},$$where $N_{accepted}$ is the number of accepted proposals.

When an appropriate reference image is selected, the NWATG statistical model successfully reconstructs images with well-preserved boundaries.
\subsection{The reconstruction of the hybrid model}\label{subsec:rec of HM}
In this section, we demonstrate the superior performance of the proposed hybrid model. Let $\mathbf{X}$ represent the CM estimates from the previous section. The goal is to construct a local regularizer $\mathbf{X}_0$ from the boundary information of $\mathbf{X}$.
While the level set method \cite{Chan2001,Li2010} is a common technique for extracting image edges, we adopt a simplified approach: a level set threshold $\tau$ is manually specified.  To obtain well-defined boundaries, we often select different thresholds for distinct regions of $\mathbf{X}$ to obtain $\mathbf{X}_0$. 
Finally, the intensity box constraint $[a,b]$ is enforced on the reconstruction to obtain the final  image $\mathbf{u}_{\text{final}}$, which is expressed as follows:$$
\mathbf{u}_{\text{final}}=\min\{\max(\mathbf{u}^{*},a),b\}.
$$ Here, $\mathbf{u}^{*}$ is the reconstruction result of the hybrid model.  For brevity, the box-constrained NWATV model is abbreviated as NWATV-box in all subsequent figures and tables. 

The experimental parameters for the box-constrained NWATV model and the proposed hybrid model are listed in Tables \ref{parameter1} and \ref{parameter2}, respectively. 

\begin{table}
	\caption{Parameters of box-constrained NWATV.}\label{parameter1}
	\begin{tabular}{l|lll}
	\hline\noalign{\smallskip}
		Reconstruction image&$\lambda$&$\rho$&$\alpha$\\
		\noalign{\smallskip}\hline\noalign{\smallskip}
		Shepp-Logan&0.006&200&5\\
		\noalign{\smallskip}\hline\noalign{\smallskip}
		Walnut&0.0001&0.4&5\\
		\noalign{\smallskip}\hline\noalign{\smallskip}
		Clinical lung&3$\times 10^{-9}$&0.001&0.01\\
	\noalign{\smallskip}\hline
	\end{tabular}
\end{table}

\begin{table}
	\caption{Parameters of the hybrid model.}\label{parameter2}
	\begin{tabular}{l|lllll}
		\hline\noalign{\smallskip}
		Reconstruction image &$\lambda$&$\rho$&$\rho_1$&$\rho_2$&$\rho_3$\\
	\noalign{\smallskip}\hline\noalign{\smallskip}
		Shepp-Logan&0.007&200&10&10&20\\
		\noalign{\smallskip}\hline\noalign{\smallskip}
		Walnut&0.0001&1&0.1&0.1&0.1\\
		\noalign{\smallskip}\hline\noalign{\smallskip}
		Clinical lung&2.5$\times 10^{-9}$&0.001&$ 10^{-5}$&$10^{-5}$&$ 10^{-5}$\\
		\noalign{\smallskip}\hline
	\end{tabular}
\end{table}

Fig. \ref{global_STI} shows the reconstruction results of the Shepp-Logan image using the box-constrained NWATV and the hybrid model, with corresponding quantitative indicators provided in Table \ref{table:hybrid_SL}. The red arrows highlight the regions of improved performance in the hybrid model reconstruction.

\begin{figure}[h]
	\centering
	\includegraphics[width=\linewidth]{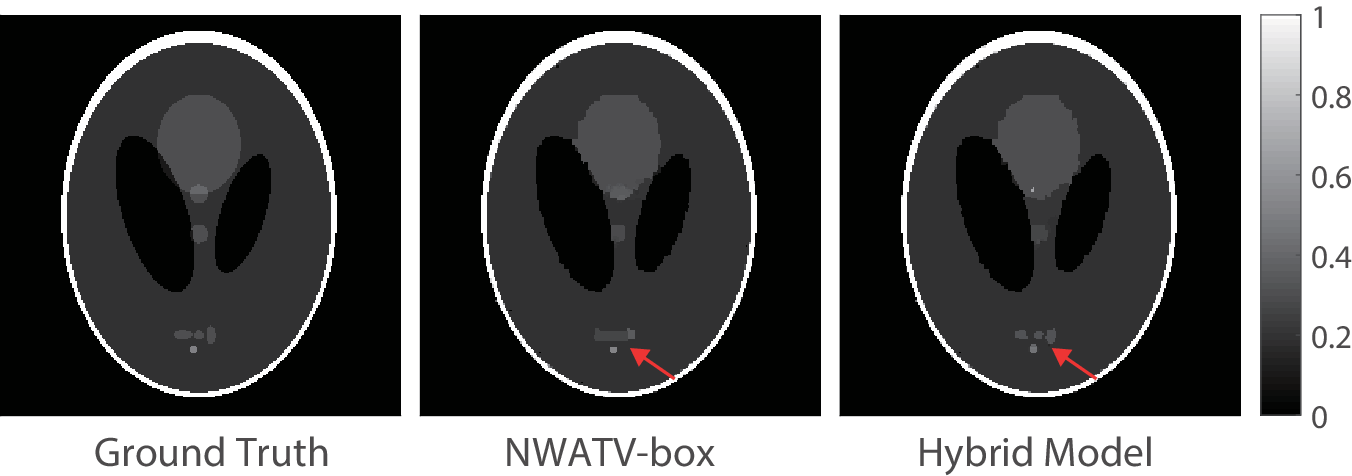}
	\caption{The reconstruction results of the Shepp-Logan image. From left to right: ground truth, the reconstruction result of the box-constrained NWATV regularization, and the reconstruction result of the hybrid model proposed in this paper.}\label{global_STI}
\end{figure}

\begin{table}
	\caption{Quantitative assessment of the box-constrained NWATV and the hybrid model on the Shepp-Logan phantom.}\label{table:hybrid_SL}
	\begin{tabular}{l|l|l|l|l|l|l|l}
		\hline\noalign{\smallskip}
		&Models&$\RE$&$\widetilde{\RE}$&$\MSE$&$\PSNR$&$\SSIM$&CPU time (s)\\
	\noalign{\smallskip}\hline\noalign{\smallskip}
		\multirow{2}{*}{ROI}&NWATV-box&0.136& 0.239&9.484$\times 10^{-4}$& 24.723&0.858&115.514\\
		\noalign{\smallskip}\cline{2-8}\noalign{\smallskip}
		&hybrid model&0.096&0.190&4.682$\times 10^{-4}$&26.051&0.952&162.558\\
		\noalign{\smallskip}\hline\noalign{\smallskip}
		\multirow{2}{*}{Global}&NWATV-box&0.066&0.121&2.651$\times 10^{-4}$&35.803&0.982&/\\
		\noalign{\smallskip}\cline{2-8}\noalign{\smallskip}
		&hybrid model&0.065&0.119&2.584$\times 10^{-4}$&35.877&0.983&/\\
	\noalign{\smallskip}\hline
	\end{tabular}
\end{table}

In the Shepp-Logan phantom reconstruction experiment, the hybrid model outperforms the NWATV-box model. Within the ROI, it reduces the $\RE$ by 29.41\% and increases the $\SSIM$ by 10.96\%. Moreover, for the global image, the hybrid model also yields improvements, decreasing the $\RE$ by 1.52\% and increasing the $\SSIM$ by 0.10\%.

Fig. \ref{Hybrid_walnut} shows the reconstruction results of the walnut image using the box-constrained NWATV and the hybrid model, with corresponding quantitative indicators provided in Table \ref{table:hybrid_walnut}. The red arrows highlight the regions of improved performance in the hybrid model reconstruction.

\begin{figure}[h]
	\centering
	\includegraphics[width=0.8\linewidth]{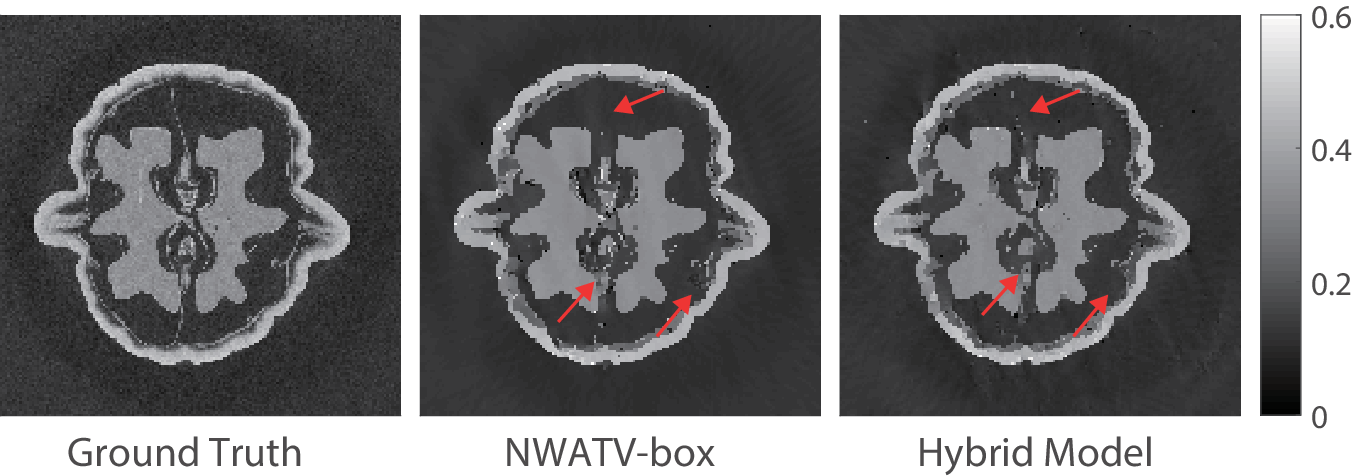}
	\caption{The reconstruction results of the walnut image. From left to right: the ground truth, the reconstruction result of the box-constrained NWATV regularization, and the reconstruction result of the hybrid model proposed in this paper.}\label{Hybrid_walnut}
\end{figure}

\begin{table}[h]
	\caption{Quantitative assessment of the box-constrained NWATV and the hybrid model on the walnut model.}\label{table:hybrid_walnut}
	\begin{tabular}{l|l|l|l|l|l|l|l}
	\hline\noalign{\smallskip}
		&Models&$\RE$&$\widetilde{\RE}$&$\MSE$&$\PSNR$&$\SSIM$&CPU time (s)\\
	\noalign{\smallskip}\hline\noalign{\smallskip}
		\multirow{2}{*}{ROI}&NWATV-box&0.312& 0.579&0.0059& 16.598&0.630&50.191\\
		\noalign{\smallskip}\cline{2-8}\noalign{\smallskip}
		&hybrid model&0.263&0.456&0.0042&15.472&0.715&272.580\\
		\noalign{\smallskip}\hline\noalign{\smallskip}
		\multirow{2}{*}{Global}&NWATV-box&0.158&0.317&0.001&25.388&0.798&/\\
		\noalign{\smallskip}\cline{2-8}\noalign{\smallskip}
		&hybrid model&0.146&0.297&8.976$\times 10^{-4}$&26.032&0.802&/\\
\noalign{\smallskip}\hline
	\end{tabular}
\end{table}

In the walnut reconstruction experiment, the hybrid model outperforms the NWATV-box model. Within the ROI, it reduces the $\RE$ by 15.71\% and increases the $\SSIM$ by 13.49\%. Moreover, for the global image, the hybrid model also yields improvements, decreasing the $\RE$ by 7.59\% and increasing the $\SSIM$ by 0.50\%. However, compared to the box-constrained NWATV result, the $\PSNR$ of the hybrid model within the ROI is 6.78\% lower. This is because the reconstruction results of the hybrid model rely on the CM estimates from sampling within the ROI, which tends to average out localized high-intensity features.

Fig. \ref{Hybrid_lung} shows the reconstruction results of the clinical lung image using the box-constrained NWATV and the hybrid model, with corresponding quantitative indicators proposed in Table \ref{table:hybrid_lung}. The red arrows highlight the regions of improved performance in the hybrid model reconstruction.

\begin{figure}[h]
	\centering
	\includegraphics[width=0.8\linewidth]{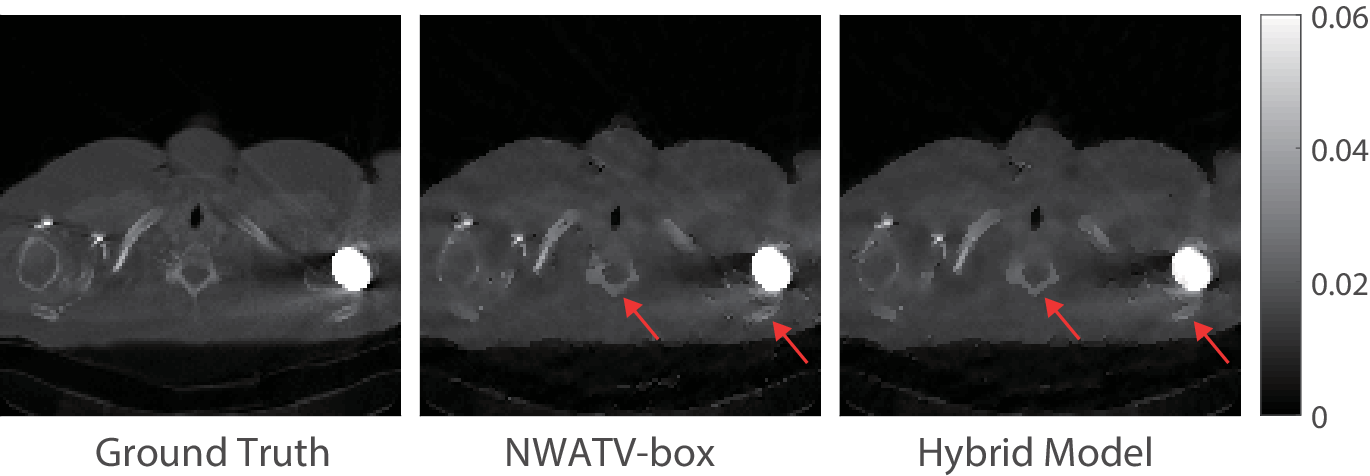}
	\caption{The reconstruction results of the clinical lung image. From left to right: ground truth, the reconstruction result of the box-constrained NWATV regularization, and the reconstruction result of the Hybrid Model proposed in this paper.}\label{Hybrid_lung}
\end{figure}

\begin{table}
	\caption{Quantitative assessment  of the box-constrained NWATV and the hybrid model on the clinical lung image.}\label{table:hybrid_lung}
	\begin{tabular}{l|l|l|l|l|l|l|l}
		\hline\noalign{\smallskip}
		&Models&$\RE$&$\widetilde{\RE}$&$\MSE$&$\PSNR$&$\SSIM$&CPU time (s)\\
		\noalign{\smallskip}\hline\noalign{\smallskip}
		\multirow{2}{*}{ROI}&NWATV-box&0.085&0.180&2.205$\times 10^{-6}$&24.666&0.998&146.728\\
	\noalign{\smallskip}\cline{2-8}\noalign{\smallskip}
		&hybrid model&0.078&0.171&1.840$\times 10^{-6}$&26.494&0.998&227.982\\
		\noalign{\smallskip}\hline\noalign{\smallskip}
		\multirow{2}{*}{Global}&NWATV-box&0.097&0.195&1.411$\times 10^{-6}$&34.419&0.998&/\\
		\noalign{\smallskip}\cline{2-8}\noalign{\smallskip}
		&hybrid model&0.091&0.179&1.233$\times 10^{-6}$&34.655&0.998&/\\
		\noalign{\smallskip}\hline
	\end{tabular}
\end{table}

In the clinical lung reconstruction experiment,  the hybrid model outperforms the NWATV-box model. Within the ROI, it reduces the $\RE$ by 8.24\%. Moreover, for the global image, the hybrid model also yields improvements, decreasing the $\RE$ by 6.19\%.

In conclusion, the experimental results demonstrate that the proposed hybrid model outperforms the traditional iterative method (the box-constrained NWATV). Specifically, under the sparse sampling condition, it achieves an average reduction of approximately 17.79\% in the $\RE$ and an average increase of 8.15\% in the $\SSIM$ within the ROI. For the global image, the corresponding improvements are approximately 5.1\% and 0.2\%, respectively. 
Compared to traditional statistical calculation methods, the hybrid model significantly reduces  the effective image dimension for sampling. This leads to a substantial decrease in the computational burden of MCMC sampling, improving the computational efficiency by approximately 97\%.

\section{Discuss}\label{sec:discuss}
In regularization-based approaches, the discrepancy principle \cite{SomersaloBook} is commonly employed to determine the regularization parameter. In statistical (Bayesian) modeling, the parameters can be treated as additional random variables within the sampling framework, thus establishing a hierarchical Bayesian model \cite{SomersaloBook} that simultaneously performs parameter estimation and image sampling.

The posterior density is defined by \eqref{J}, which includes a regularization term $\mathrm{R}(\mathbf{u}_{\mathrm{ROI}})=\lambda_{\mathrm{ROI}}\|\mathbf{p}_{\mathrm{ROI}}\cdot\mathbf{\mathcal{D}}\mathbf{u}_{\mathrm{ROI}}\|_{\ell_{1}}$. However, parameters selection with this posterior is challenging because its normalization constant is difficult to compute. To leverage the tractable normalization constant of the Gaussian distribution, we introduce a parameter $\delta$ (where $\lambda_{\mathrm{ROI}}\delta = 1$) into the data fidelity term. Subsequently, we treat $\delta$ itself as a parameter to be inferred. We assign a Gamma distribution prior to the parameter $\delta$. Its probability density function is given by \cite{Lawless2003}:
$$G(x;a,b)=\frac{b^ax^{a-1}}{\Gamma(a)}\exp(-bx),\quad x>0.$$Here, $a$ is the shape parameter governing the form of the distribution, and $b^{-1}$ is the scale parameter controlling its dispersion. For $a \geq 1$, the peak of the distribution is located at $(a-1)/b$.

Consequently,  $\delta$ satisfies$$
\delta\sim\delta^{\alpha_{\delta}-1}\exp\{-\beta_{\delta}\delta\},
$$and we choose $\alpha_{\delta}=1,\beta_{\delta}=10^{-8}$ to ensure the prior of $\delta$ provides almost no information. Then, the posterior density can be written as$$
\frac{d\nu^y_n}{d\nu_0}\propto\delta^{M/2+\alpha_{\delta}-1}\exp\{-\frac{\lambda_{\mathrm{ROI}}\delta}{2\sigma^2}\|\mathbf{y}_{\mathrm{ROI}}-A_{\mathrm{ROI}}\mathbf{u}_{\mathrm{ROI}}\|_{\ell_{2}}^2-\lambda_{\mathrm{ROI}}\|\mathbf{p}_{\mathrm{ROI}}\cdot\mathbf{\mathcal{D}}\mathbf{u}_{\mathrm{ROI}}\|_{\ell_{1}}-\beta_{\delta}\delta\}.
$$Choosing a Gamma prior for $\delta$ ensures that its conditional posterior distribution is also a Gamma distribution, which is straightforward to sample from. Therefore, in the $k$-th sampling iteration, $\delta^{(k)}$ is drawn from its conditional Gamma distribution\begin{equation}
	\delta^{(k)}\sim(\delta^{(k)})^{M/2+\alpha_{\delta}-1}\exp\left\{-\left(\frac{\lambda_{\mathrm{ROI}}^{(k-1)}}{2\sigma^2}\|\mathbf{y}_{\mathrm{ROI}}-A_{\mathrm{ROI}}\mathbf{u}_{\mathrm{ROI}}^{(k-1)}\|_{\ell_{2}}^{2}+\beta_{\delta}\right)\delta^{(k)}\right\}.\label{ga}
\end{equation}Because $\beta_{\delta}$ is sufficiently small, the maximum a posteriori (MAP) estimate for $\delta$ at the $k$-th iteration is given by$$
\delta_{\mathrm{MAP}}^{(k)}=\frac{M/2+\alpha_{\delta}-1}{\frac{\lambda_{\mathrm{ROI}}^{(k-1)}}{2\sigma^2}\|\mathbf{y}_{\mathrm{ROI}}-A_{\mathrm{ROI}}\mathbf{u}_{\mathrm{ROI}}^{(k-1)}\|_{\ell_{2}}^2+\beta_{\delta}}\approx{\frac{M/2+\alpha_{\delta}-1}{\frac{\lambda_{\mathrm{ROI}}^{(k-1)}}{2\sigma^2}\|\mathbf{y}_{\mathrm{ROI}}-A_{\mathrm{ROI}}\mathbf{u}_{\mathrm{ROI}}^{(k-1)}\|_{\ell_{2}}^2}}.
$$We assume that  $\delta_{\mathrm{MAP}}^{(k)}$ varies within a narrow range. Namely, the conditional posterior of $\delta^{(k)}$ is a Gamma distribution with a small variance, concentrating its probability mass near $\delta_{\mathrm{MAP}}^{(k)}$. Furthermore, the corresponding sequence ${ \lambda_{\mathrm{ROI}}^{(k)} }$ also exhibits limited variation. Based on this assumption of limited fluctuation,  we set the initial value as $\lambda_{\mathrm{ROI}}^{(0)}=10^{-3}$ and compute the corresponding initial $\delta^{(0)}$ using its MAP estimator form: $$\delta^{(0)}={\frac{M/2+\alpha_{\delta}-1}{\frac{\lambda_{\mathrm{ROI}}^{(0)}}{2\sigma^2}\|\mathbf{y}_{\mathrm{ROI}}-A_{\mathrm{ROI}}\mathbf{u}_{\mathrm{ROI}}^{(0)}\|_{\ell_{2}}^2}}.$$ It then follows that the ratio $\frac{\delta^{(k)}}{\lambda^{(k)}}$ remains approximately constant across iterations. We denote this approximate constant by $T$: \begin{equation}
	\frac{\delta^{(k)}}{
		\lambda_{\mathrm{ROI}}^{(k)}}\approx\frac{\delta^{(0)}}{\lambda_{\mathrm{ROI}}^{(0)}}=:T.\label{T}
\end{equation}
Then, \begin{equation}
	\lambda_{\mathrm{ROI}}^{(k)}=\frac{1}{\delta^{(k)}}\approx\frac{1}{\delta_{MAP}^{(k)}}=\frac{\frac{\lambda_{\mathrm{ROI}}^{(k-1)}}{2\sigma^2}\|\mathbf{y}_{\mathrm{ROI}}-A_{\mathrm{ROI}}\mathbf{u}_{\mathrm{ROI}}^{(k-1)}\|_{\ell_{2}}^2}{M/2+\alpha_{\delta}-1},\label{lam}
\end{equation}and we have $$\frac{\lambda_{\mathrm{ROI}}^{(k)}}{\lambda_{\mathrm{ROI}}^{(k-1)}}\approx\frac{\frac{1}{2\sigma^2}\|\mathbf{y}_{\mathrm{ROI}}-A_{\mathrm{ROI}}\mathbf{u}_{\mathrm{ROI}}^{(k-1)}\|_{\ell_{2}}^2}{M/2+\alpha_{\delta}-1}.$$If the initial values are poor and the image region is small (such as an ROI), the iterative process can lead to a situation where $$\frac{1}{2\sigma^2}\|\mathbf{y}_{\mathrm{ROI}}-A_{\mathrm{ROI}}\mathbf{u}_{\mathrm{ROI}}^{(k-1)}\|_{\ell_{2}}^2>M/2+\alpha_{\delta}-1.$$ Consequently, the value of $\lambda_{\mathrm{ROI}}$ increases monotonically, as illustrated in Fig. \ref{choose}(a), eventually reaching an order of magnitude as high as $10^{300}$. Under these conditions,  the regularization parameter $\lambda_{\mathrm{ROI}}^{(k)}$ fails to balance the data fidelity and regularization terms when sampling $\mathbf{u}_{\mathrm{ROI}}^{(k)}$.

To establish a stable scheme for selecting $\lambda_{\mathrm{ROI}}$, we leverage the approximate relation \eqref{T} to derive the following approximation of \eqref{lam}:$$
\lambda_{\mathrm{ROI}}^{(k)}=\sqrt{\frac{\lambda_{\mathrm{ROI}}^{(k-1)}\delta^{(k)}\|\mathbf{y}_{\mathrm{ROI}}-A_{\mathrm{ROI}}\mathbf{u}_{\mathrm{ROI}}^{(k-1)}\|_{\ell_{2}}^2}{2\sigma^2T(M/2+\alpha_{\delta}-1)}}.
$$As shown in Fig. \ref{choose} (b), the parameter $\lambda_{\mathrm{ROI}}$ takes values approximately within the range $[1.12\times10^{-3},1.24\times10^{-3}]$. In practice, since $\delta$ is sampled stochastically, the product $\lambda_{\mathrm{ROI}}\delta$ fluctuates accordingly. 
\begin{figure}[h]
	\centering
	\subfloat[The sampling process of $\lambda_{\mathrm{ROI}}$ when  $\lambda_{\mathrm{ROI}}=1/\delta_{\mathrm{MAP}}$.]{
		\begin{minipage}{0.4\linewidth}
			\centering
			\includegraphics[height=0.8\linewidth]{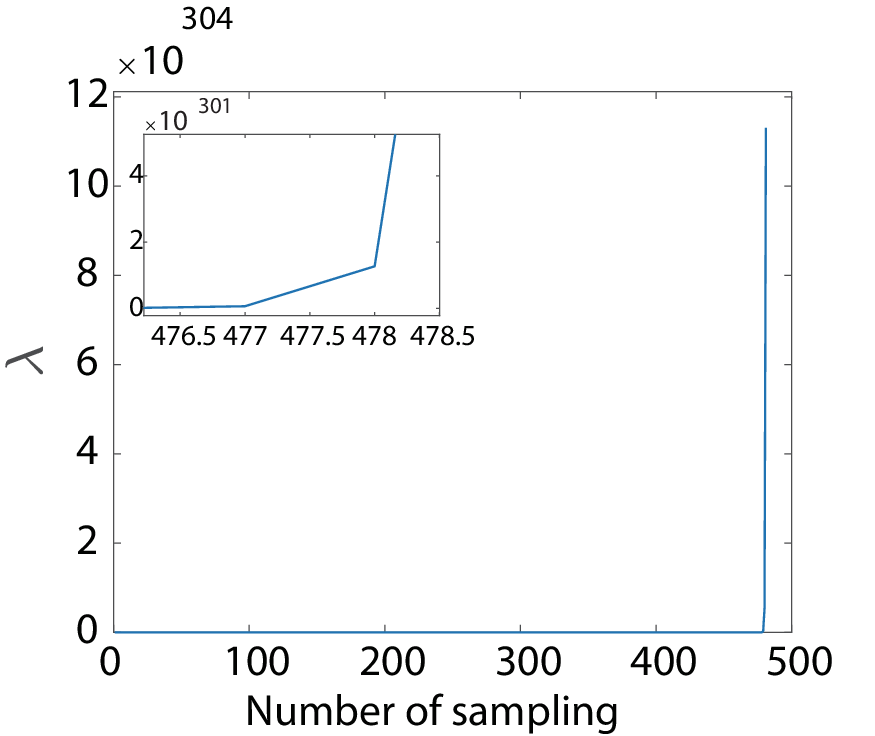}
	\end{minipage}}
	~
	\subfloat[The density function of $\lambda_{\mathrm{ROI}}$ when $\lambda_{\mathrm{ROI}}=1/\delta_{\mathrm{MAP}}$ is modified by $T$.]{
		\begin{minipage}{0.4\linewidth}
			\centering
			\includegraphics[height=0.8\linewidth]{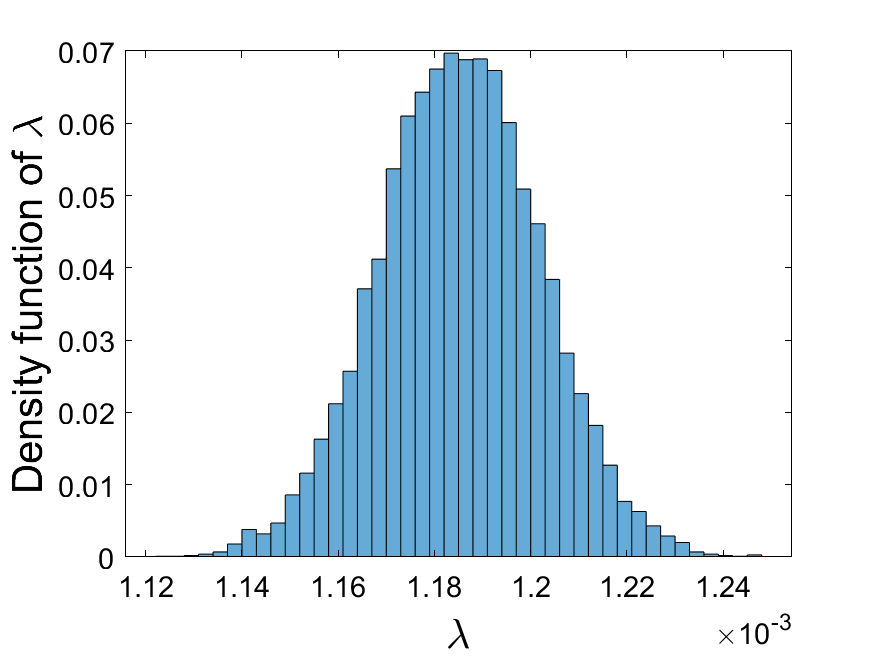}
	\end{minipage}}
	\caption{Sampling of $\lambda_{\mathrm{ROI}}$ under the constraint $\lambda_{\mathrm{ROI}} = 1 / \delta_{\mathrm{MAP}}$ with step size $\alpha = 0.002$. (a) shows the evolution of $\lambda_{\mathrm{ROI}}$ across sampling iterations; (b) represents the frequency histogram of $\lambda_{\mathrm{ROI}}$ with the modified constant $T=\delta^{(0)}/\lambda_{\mathrm{ROI}}^{(0)}$.}\label{choose}
\end{figure}

To make the value of $\lambda_{\mathrm{ROI}}^{(k)}\delta^{(k)}$ closer to 1, we correct $\delta_{\mathrm{MAP}}^{(k)}$ using $\lambda_{\mathrm{ROI}}^{(k)}\delta^{(k)}$. The specific formula is as follows:
\begin{equation}
	\lambda_{\mathrm{ROI}}^{(k)}\cdot(\delta_{\mathrm{MAP}}^{(k)}\cdot\lambda_{\mathrm{ROI}}^{(k)}\delta^{(k)})=1.\label{lambda_delta}
\end{equation}

Next, we employ the MCMC method\cite{SomersaloBook}  to to sample from the probability distribution of $\delta$. This sampling is based on the energy function $$J(\mathbf{u}_{\mathrm{ROI}})=\frac{\lambda_{\mathrm{ROI}}\delta}{2\sigma^2}\|\mathbf{y}_{\mathrm{ROI}}-A_{\mathrm{ROI}}\mathbf{u}_{\mathrm{ROI}}\|_{\ell_{2}}^2+\lambda_{\mathrm{ROI}}\|\mathbf{p}_{\mathrm{ROI}}\cdot\mathbf{\mathcal{D}}\mathbf{u}_{\mathrm{ROI}}\|_{\ell_{1}},$$which is derived from the posterior distribution. The complete sampling procedure is outlined in Algorithm \ref{algorithm:delta}.
\begin{algorithm}[h]
	\caption{Hierarchical model of parameter $\delta$.}
	\label{algorithm:delta}
	\begin{algorithmic}[1]
		\REQUIRE {Projection matrix $A_{\mathrm{ROI}}$, the measured data $\mathbf{y}_{\mathrm{ROI}}$, the number of projection lines $M$, the standard deviation of noise $\sigma$. \\\ \ \ \ Parameters: $\alpha_{\delta},\beta_{\delta},\beta\in\mathbb{R}^{+},$ the maximum number of iterations kmax$\in\mathbb{Z}^{+}$.}
		
		\ENSURE {Parameter $\delta$ and the CM estimate $\mathbf{X}$.}
		\STATE {Initial value: $\mathbf{u}_{\mathrm{ROI}}^{(0)},\lambda_{\mathrm{ROI}}^{(0)},\delta^{(0)}$.}
		\FOR{k=1:kmax}
		\STATE {Draw $\delta^{(k)}$ from \eqref{ga}.}
		\STATE{From \eqref{lambda_delta}, $\lambda_{\mathrm{ROI}}^{(k)}=\sqrt[4]{\frac{\lambda_{\mathrm{ROI}}^{(k-1)}\delta^{(k)}\|\mathbf{y}_{\mathrm{ROI}}-A_{\mathrm{ROI}}\mathbf{u}_{\mathrm{ROI}}^{(k-1)}\|_{\ell_{2}}^2}{2\sigma^2T^2(M/2+\alpha_{\delta}-1)}}$.}
		\STATE {Update $\mathbf{v}=\sqrt{1-\alpha^2}\mathbf{u}_{\mathrm{ROI}}+\alpha\mathbf{w}$, where $\mathbf{w}\sim\mathcal{N}(0,C)$.}
		\STATE {Draw $t\sim\mathcal{U}[0,1]$.}
		\IF{$t\le\min\{1,\exp\{J(\mathbf{u}_{\mathrm{ROI}}^{(k)})-J(\mathbf{v})\}\}$} 
		\STATE{$\mathbf{u}_{\mathrm{ROI}}^{(k)}=\mathbf{v}$;}
		\ELSE 
		\STATE{$\mathbf{u}_{\mathrm{ROI}}^{(k)}=\mathbf{u}_{\mathrm{ROI}}^{(k-1)}$;}\ENDIF
		\ENDFOR
		\STATE{Calculate the conditional mean $\mathbf{X}$.}
	\end{algorithmic}
\end{algorithm}
Specifically, we use the Shepp-Logan phantom shown in Fig. \ref{figure:ground}(a) with the following settings: noise level = 1\%, number of projection angles = 60, step size = 0.002, and total number of MCMC samples = 10,000. The evolution of the log-posterior with sample number, plotted in Fig. \ref{parameter_J}(a), allows us to estimate the required burn-in period. In fact, the log-posterior density stabilizes within a narrow range, indicating that the burn-in period is negligible. Furthermore, Fig. \ref{histogram} presents the frequency histograms for $\delta$ and $\lambda_{\mathrm{ROI}}$. The histogram for $\delta$ confirms that it follows a narrow Gamma distribution, while \ref{histogram}(b) shows that $\lambda_{\mathrm{ROI}}$ lies approximately within $[1.15 \times 10^{-3}, 1.22 \times 10^{-3}]$. Given that $\lambda_{\mathrm{ROI}}$ also exhibits a narrow distribution, we select its MAP estimate $\lambda_{\mathrm{MAP}}$, as the final regularization parameter. The corresponding CM estimate with parameter $\lambda_{\mathrm{MAP}}$ is displayed in Fig. \ref{lambda_map}. For this sampling process, the step size is 0.002, the first 8000 samples are discarded as burn-in period, and the acceptance rate is 28.30\%.
\begin{figure}[hbt]
	\centering
	\subfloat[The evolution of $-J$ in the algorithm \ref{algorithm:delta}.]{
		\begin{minipage}{0.4\linewidth}
			\centering
			\includegraphics[width=\linewidth]{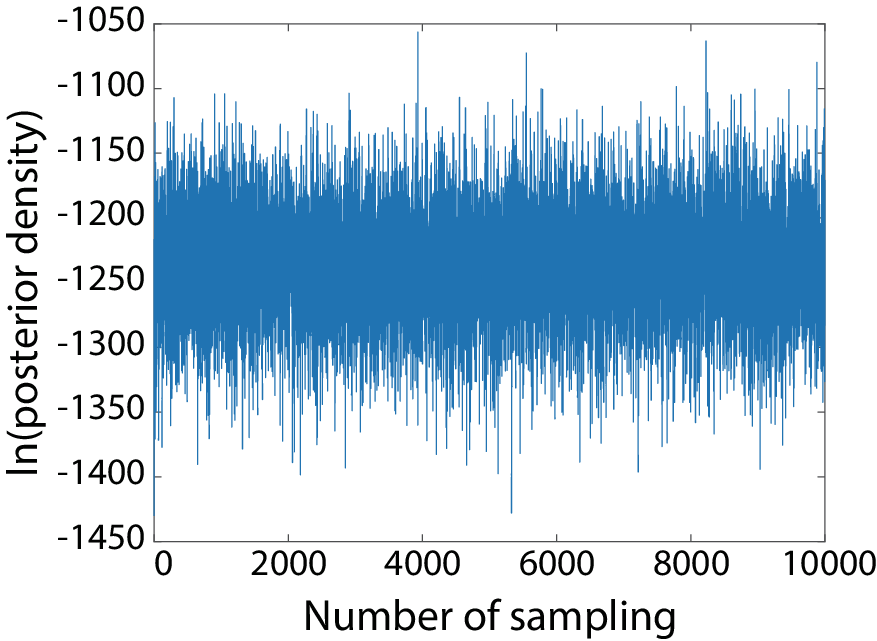}
	\end{minipage}}
	~
	\subfloat[The evolution of $\lambda_{\mathrm{ROI}}$ in the algorithm \ref{algorithm:delta}.]{
		\begin{minipage}{0.4\linewidth}
			\centering
			\includegraphics[width=\linewidth]{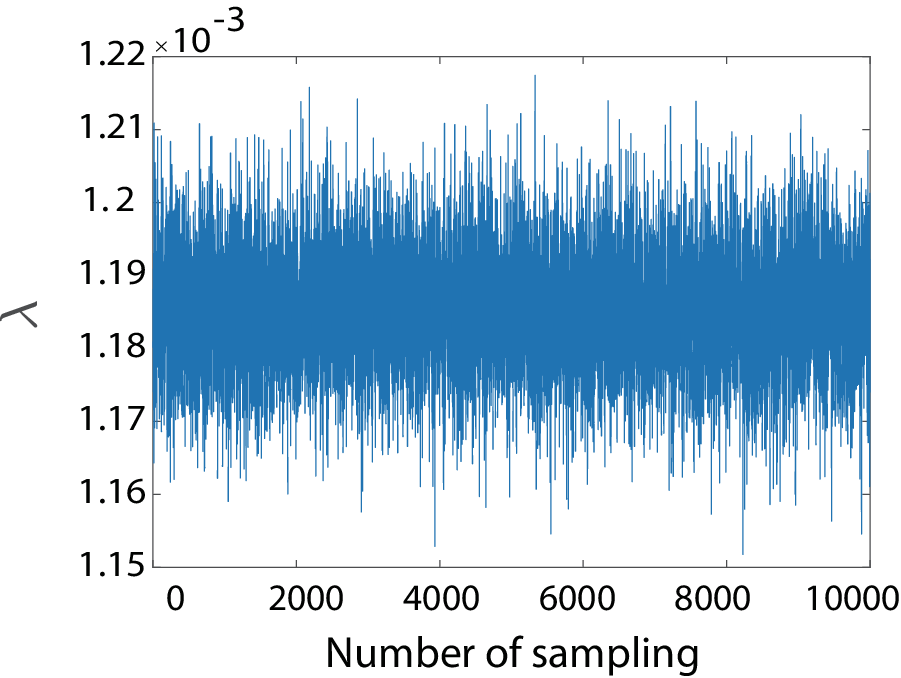}
	\end{minipage}}
	\caption{The evolution with 10000 samples.}
	\label{parameter_J}
\end{figure}
\begin{figure}[h]
	\centering
	\subfloat[Density function of parameter $\delta$.]{
		\begin{minipage}{0.4\linewidth}
			\centering
			\includegraphics[width=\linewidth]{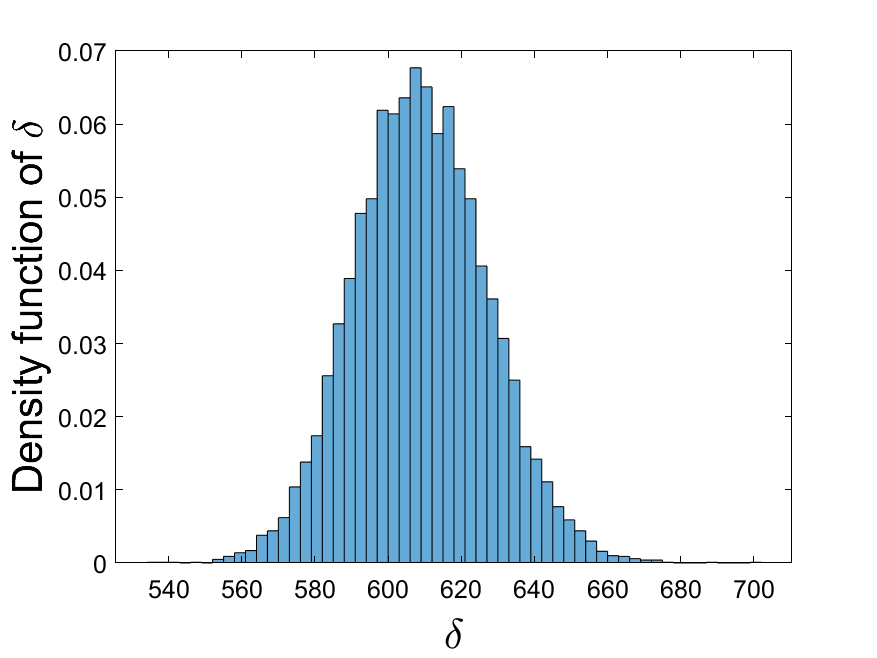}
	\end{minipage}}
	~
	\subfloat[Density function of the parameter $\lambda_{\mathrm{ROI}}$.]{
		\begin{minipage}{0.4\linewidth}
			\centering
			\includegraphics[width=\linewidth]{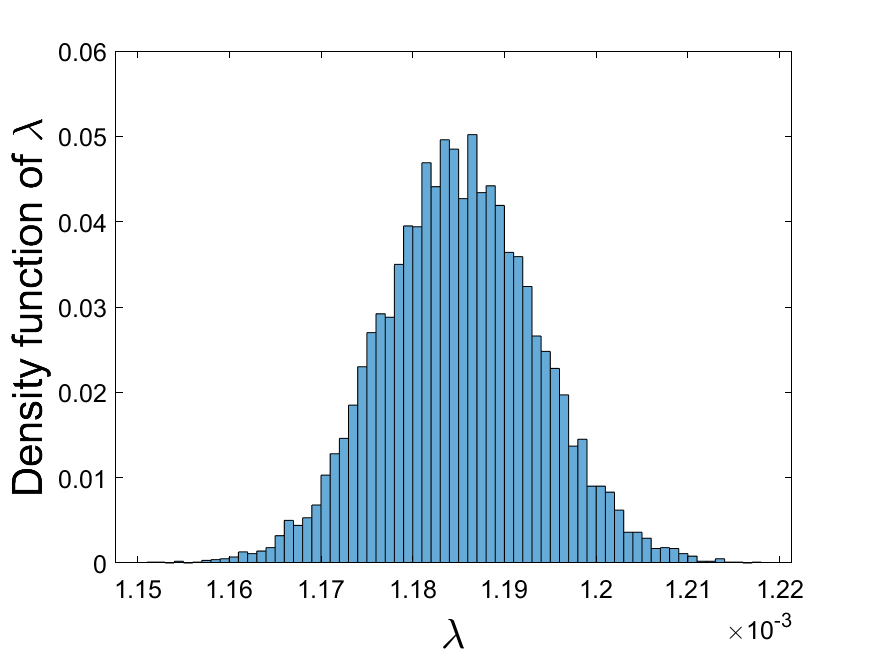}
	\end{minipage}}
	\caption{Frequency histograms of corresponding parameters.  (a) represents the distribution of $\delta$; (b) represents the distribution of $\lambda_{\mathrm{ROI}}$.}\label{histogram}
\end{figure}
\begin{figure}[h]
	\centering
	\includegraphics[width=0.6\linewidth]{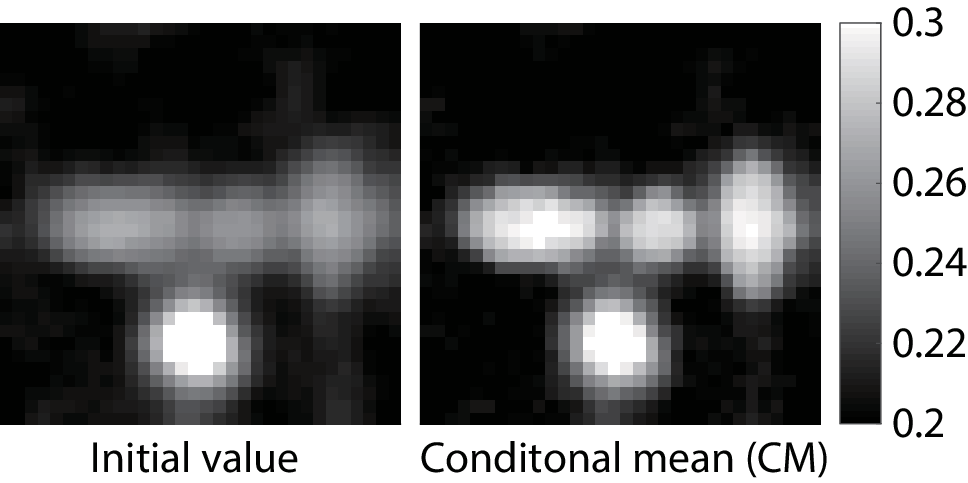}
	\caption{The CM estimate with the final parameter $\lambda_{\mathrm{MAP}}$. The step size is $\alpha=0.002$. The left is the initial value and the right is the CM estimate.}\label{lambda_map}
\end{figure}
The trace confirms that $\lambda_{\mathrm{ROI}}$ remains stable within a narrow range throughout sampling, and its precise level exhibits a dependence on the chosen initial value.
\section{ Conclusions and future works}\label{sec:conclusion}
In this paper, we have proposed a hybrid model that integrates statistical and traditional regularization approaches for sparse-view CT reconstruction. The superiority of the proposed model was demonstrated across multiple datasets, including the Shepp-Logan phantom, actual walnut projections, and the clinical lung CT image. The model effectively improves reconstruction quality, notably by preserving fine structural details and sharp image boundaries.

In the future, we  plan to explore the use of the Segment Anything Model (SAM)\cite{Kirillov2023} for automated boundary extraction.

	\end{sloppypar}
\end{document}